\documentclass[aps, prx, reprint, superscriptaddress, longbibliography]{revtex4-2}

\usepackage{xcolor}
\usepackage[normalem]{ulem}
\usepackage{soul}

\usepackage{amsmath,amssymb,amsfonts,amsthm,dsfont}

\usepackage{tikz}
\usetikzlibrary{positioning,arrows.meta,calc}

\usepackage{pgfplots}
\pgfplotsset{compat=1.18}

\usepackage{verbatim}

\usepackage{hyperref}

\newtheorem{theorem}{Theorem}
\newtheorem{definition}[theorem]{Definition}
\newtheorem{lemma}[theorem]{Lemma}
\newtheorem{proposition}[theorem]{Proposition}
\newtheorem{corollary}[theorem]{Corollary}

\begin{document}


\def \hattheta{\hat{\theta}}
\def \hata{\hat{\theta}}
\def \post{\mathrm{post}}

\def \Cbb{\mathbb{C}}
\def \Ebb{\mathbb{E}}
\def \Hbb{\mathbb{H}}
\def \Jbb{\mathbb{J}}
\def \Lbb{\mathbb{L}}
\def \Nbb{\mathbb{N}}
\def \Rbb{\mathbb{R}}
\def \Sbb{\mathbb{S}}

\def \Ccal{\mathcal{C}}
\def \Ecal{\mathcal{E}}
\def \Hcal{\mathcal{H}}
\def \Lcal{\mathcal{L}}

\def\csf{\mathsf{c}}
\def\ssf{\mathsf{s}}
\def\Bsf{\mathsf{B}}
\def\Csf{\mathsf{C}}
\def\Fsf{\mathsf{F}}
\def\Hsf{\mathsf{H}}
\def\Isf{\mathsf{I}}
\def\Jsf{\mathsf{J}}
\def\Ksf{\mathsf{K}}
\def\Msf{\mathsf{M}}
\def\Rsf{\mathsf{R}}
\def\Vsf{\mathsf{V}}
\def\Wsf{\mathsf{W}}
\def\Zsf{\mathsf{Z}}

\def \Brm{\mathrm{B}}
\def \cfBLD{{\cal C}^f_\mathrm{B}}
\def \clamBLD{{\cal C}^{(\la)}_\mathrm{B}}
\def \cvT{{\cal C}_\mathrm{vT}}
\def \cfvT{{\cal C}^f_\mathrm{vT}}
\def \cNH{{\cal C}_\mathrm{BNH}}
\def \cH{{\cal C}_\mathrm{H}}
\def \cN{{\cal C}_\mathrm{N}}
\def \JvT{\mathsf{J}_\mathrm{vT}}
\def \Vmin{\mathsf{V}_\mathrm{min}}
\def \VsfB{\mathsf{V}_{\mathrm{B}}}
\def \rhoB{\rho_{\mathrm{B}}}

\newcommand{\Tr}{\mathrm{Tr}}
\newcommand{\TrAbs}{\mathrm{TrAbs}}
\newcommand{\trAbs}{\mathrm{trAbs}}
\newcommand{\tr}{\mathrm{tr}}
\newcommand{\Ttr}{\mathds{T}\mathrm{r}}
\newcommand{\overl}{\overline}
\newcommand{\inv}{{-1}}
\newcommand{\Hil}{\mathcal{H}}
\newcommand{\ta}{\theta}
\newcommand{\Real}{\mathrm{Re}\,}
\newcommand{\Imag}{\mathrm{Im}\,}
\newcommand{\del}{\partial}

\newcommand{\R}{\mathds{R}}
\newcommand{\C}{\mathds{C}}
\newcommand{\F}{\mathcal{F}}
\newcommand{\Xset}{\mathfrak{X}}
\newcommand{\X}{\mathcal{X}}
\newcommand{\Li}{\mathcal{L}}
\newcommand{\FS}{\F(\Sp)}
\newcommand{\Sds}{\mathds{S}}
\newcommand{\Xds}{\mathds{X}}
\newcommand{\Ads}{\mathds{A}}
\newcommand{\Lds}{\mathds{L}}


\newcommand{\udef}{\underline{Def}}
\newcommand{\lld}{\lambda LD}
\newcommand{\M}{\mathcal{M}}
\newcommand{\D}{\mathcal{D}}
\newcommand{\Til}{\Tilde}
\newcommand{\E}{\mathcal{E}}

\newcommand{\ket}[1]{|#1 \rangle}
\newcommand{\bra}[1]{\langle #1|}
\newcommand{\braket}[2]{\langle #1|#2 \rangle}
\newcommand{\ketbra}[2]{\ket{#1}\bra{#2}}

\newcommand{\inner}[2]{\langle #1,#2 \rangle}

\newcommand{\plam}{\frac{1+\lambda}{2}}
\newcommand{\mlam}{\frac{1-\lambda}{2}}
\def \hPi{\hat{\Pi}}
\def \ot{\otimes}
\def \Xcal{{\cal X}}
\def \la{\lambda}
\def \sq{\sqrt}
\def \ep{\epsilon}
\def \half{\frac{1}{2}}

\newcommand{\AffUEC}{
 \textit{Graduate School of Informatics and Engineering, The University of Electro-Communications,}\\
Tokyo, 182-8585 Japan}

\title{Bayesian Monotone Metrics for Multiparameter Quantum Estimation}

\author{Jianchao Zhang}
\email[]{jianchaozhang01@gmail.com}
\affiliation{\AffUEC}
\author{Koichi Yamagata}
\email[]{yamagata@se.kanazawa-u.ac.jp}
\affiliation{Institute of Science and Engineering, Kanazawa University \\
Kanazawa, Ishikawa, 920-1192, Japan}
\author{Jun Suzuki}
\email[]{junsuzuki@uec.ac.jp}
\affiliation{\AffUEC}

\date{\today}

\begin{abstract}
Bayesian quantum estimation offers a finite-data framework for quantum sensing and metrology, yet a unified geometric formulation for multiparameter Bayes risk has been lacking. 
We introduce \emph{Bayesian monotone metrics} by evaluating Petz monotone metrics on the prior-averaged state, providing a Bayesian extension of the full class of statistically meaningful (CPTP) quantum metrics. 
This framework yields Bayesian quantities, including quantum posterior-mean operators and a quantum Bayesian dual Fisher-information matrix, and it leads to a systematic family of computable lower bounds on the Bayes risk.
The resulting bounds naturally incorporate multiparameter measurement incompatibility and, for every monotone metric in the family, we prove a universal dominance over the corresponding quantum van Trees (Bayesian Cram\'er--Rao) bound.  
Moreover, we show that optimizing over all operator monotone functions collapses to a one-parameter subfamily, turning the tightest bound into a tractable optimization with a clear geometric interpretation. 
In representative examples, the optimized bounds are strictly tighter than the Bayesian SLD and RLD bounds.
Our results establish Bayesian monotone metrics as a unifying information-geometric perspective on Bayesian quantum estimation, enabling systematic and computable performance limits in multiparameter settings.
\end{abstract}


\maketitle

\section{Introduction}
Estimating unknown parameters encoded in quantum states or processes is a central task in quantum sensing and quantum metrology \cite{caves1981quantum,giovannetti2004quantum,demkowicz2012elusive,giovannetti2011,degen17,szczykulska2017reaching,proctor2018multiparameter,pirandola2019fundamental}. 
Geometric viewpoints that relate local distinguishability to metric structures have become a standard language in quantum estimation \cite{uhlmann1976transition,wootters1981statistical,hubner1992explicit,braunstein1994statistical,petz1996geometries,ANbook,bengtsson2017geometry,sidhu2020geometric}.
In realistic finite-data settings, performance is ultimately determined not only by the measurement and the probe design, but also by prior information and the decision rule used to process outcomes \cite{helstrom1969quantum,van2004detection,escher2011general,meyer2025quantum}.
Accordingly, Bayesian quantum estimation offers an operational framework to incorporate prior knowledge and to assess performance via the Bayes risk, making it a natural language for adaptive and sequential protocols  \cite{berry2000optimal,higgins2007entanglement,wiseman2009book,granade2012robust,mahler2013adaptive,dinani2019bayesian}.

Despite substantial progress on quantum limits for point estimation, the Bayesian multiparameter setting lacks a comparably systematic geometric formulation \cite{helstrom1976,holevobook,albarelli2025measurement,rubio2020bayesian,tsang2020physics,barndorff2000fisher,gebhart2021bayesian}.
In point estimation, Petz's theory of monotone metrics provides a principled family of contractive Riemannian metrics on the quantum-state manifold \cite{petz1996monotone,petz1996geometries,lesniewski1999monotone,cencov2000statistical,Petz2010introduction}.
This framework organizes logarithmic derivatives, Fisher-information-type objects, and fundamental performance limits, while making explicit the role of multiparameter measurement incompatibility in geometric terms \cite{holevobook,holevo1977commutation,GM00,matsumoto2002new,razavian2020quantumness,chen2022incompatibility,hayashi2023tight}.
In Bayesian estimation, however, the natural objects are no longer local quantities defined at a fixed $\theta$, but prior-averaged states and Bayes risk as the optimization objective, so it is not immediately clear how to transfer the advantages of monotone-metric geometry to this setting \cite{GillLevit1995,van2004detection}. 
At the same time, the Bayesian formulation brings additional structure (e.g., integration-by-parts identities under mild regularity conditions), which has been exploited to derive Bayesian performance limits \cite{jupp2010van,tsang2020physics}.
Recent progress has also produced finite-sample multiparameter limits and their Bayesian counterparts, including semidefinite-programming formulations \cite{albarelli2019evaluating,gorecki2020pi,conlon2021efficient,suzuki2024bayesian}.
These considerations raise a natural question of whether Petz’s monotone-metric geometry can be extended to Bayesian multiparameter estimation in a tractable way that yields explicit Bayes-risk lower bounds.

Motivated by this gap, we introduce Bayesian monotone metrics by evaluating Petz monotone metrics at the prior-averaged state.
Geometrically, the prior-averaged state serves as a reference point in the state space, at which we evaluate a monotone metric. This imports the information-geometric machinery of quantum point estimation into the Bayesian setting while keeping the Bayes risk at center stage.
Moreover, since Petz monotone metrics are in one-to-one correspondence with operator monotone functions~\cite{petz1996monotone}, varying the metric yields a systematic family of metric-induced Bayesian lower bounds for the Bayes risk.
In the classical Bayesian setting, posterior means and related covariance identities are built from ordinary multiplication and pointwise division, so their probabilistic meaning is immediate. In the quantum setting, however, such operations do not admit a unique noncommutative analogue. The relevant question is therefore not how to define an arbitrary quantum analogue, but how to identify those extensions that remain statistically meaningful, i.e., compatible with data processing under CPTP maps. Petz monotone metrics provide exactly this class in quantum point estimation, and our construction extends it to Bayesian estimation through the prior-averaged state.
The construction is particularly well behaved in Bayesian problems: averaging over the prior can regularize rank-deficient models and typically yields a full-rank averaged state, which stabilizes the resulting metric-induced objects and makes them unambiguous.

Building on this metric viewpoint, we define quantum posterior-mean operators and a corresponding Bayesian dual Fisher-information matrix \cite{helstrom1976,holevobook}.
In particular, a chosen metric specifies quantum posterior-mean operators \cite{personick1971application} via a Bayesian posterior-mean equation, and it induces a quantum Bayesian dual Fisher-information matrix $\Ksf_\Brm$ as the associated Gram matrix.
This systematically integrates a quantum posterior mean operator, defined via Petz monotone metrics at the prior-averaged state, into a geometric framework for multiparameter Bayes-risk bounds.
These objects provide an information-geometric interpretation of Bayesian uncertainty: the second-moment matrix $\Msf$ decomposes into a metric-induced information term and a remainder $\Msf-\Ksf_\Brm$, which we interpret as a quantum posterior variance matrix.
From this structure we obtain a family of Bayes-risk lower bounds (quantum posterior variance bound), in which multiparameter measurement incompatibility enters through an intrinsic metric-dependent term rather than an ad hoc correction.
Moreover, for every monotone metric in the family we establish a universal dominance relation over the corresponding quantum van Trees (quantum Bayesian Cram\'er--Rao) bound \cite{personick1970efficient,personick1971application}.

Given this family and its dominance over van Trees, a natural next question is how to identify the tightest bound within this full metric-induced family: in principle one would need to optimize over all operator monotone functions, which is an infinite-dimensional problem.
Here our geometric construction provides a substantial simplification.
We show that it suffices to consider a one-parameter $\lambda$-dependent subfamily \cite{yamagata2021maximum}, and moreover, for any operator monotone function there exists a corresponding $\lambda$ in this subfamily that yields an equal or tighter bound.
This reduction turns the search for the tightest posterior variance bound into a tractable single-parameter optimization.

Beyond this structural simplification, we demonstrate through explicit qubit examples that the resulting posterior variance bounds can be strictly tighter than the Bayesian SLD \cite{helstrom1967minimum} and Bayesian RLD \cite{YL73} bounds, with a clear quantitative gap.
This demonstrates that importing monotone-metric geometry into the Bayesian setting is not merely a reformulation: the Bayesian monotone-metric family reveals performance limits that are genuinely stronger than those obtained from the standard SLD/RLD prescriptions.

Our main contributions are as follows:
\begin{enumerate}
\item[(i)] We introduce Bayesian monotone metrics by evaluating Petz monotone metrics at the prior-averaged state (Definition~\ref{def:Bmonotone}), providing a principled Bayesian extension of the standard class of statistically meaningful (CPTP) quantum metrics.

\item[(ii)] This metric viewpoint yields a systematic family of Bayes-risk lower bounds, indexed by operator monotone functions that are in one-to-one correspondence with monotone metrics. This construction introduces the quantum posterior-mean operator and the quantum Bayesian dual Fisher information (Definition~\ref{def:QPM-QBFI}); it then yields the quantum posterior variance inequality (Theorem~\ref{thm:fBLD}) and its weighted Bayes-risk form (Corollary~\ref{cor:fBLD}), as well as the Bayesian $\lambda$-posterior variance bound (Definition~\ref{def:BlaLD}).

\item[(iii)] We prove that, for every metric in this family, the resulting bound universally dominates the corresponding quantum van Trees (quantum Bayesian Cram\'er--Rao) bound; see Theorem~\ref{thm:ineq2}.

\item[(iv)] We show that optimizing over the full metric family reduces to a tractable one-parameter $\lambda$-subfamily; see Theorem~\ref{thm:ineq1} and Corollary~\ref{cor:ineq1}. Additionally, numerical examples in Sec.~\ref{sec:example} exhibit strict improvements over the Bayesian SLD and RLD bounds.
\end{enumerate}

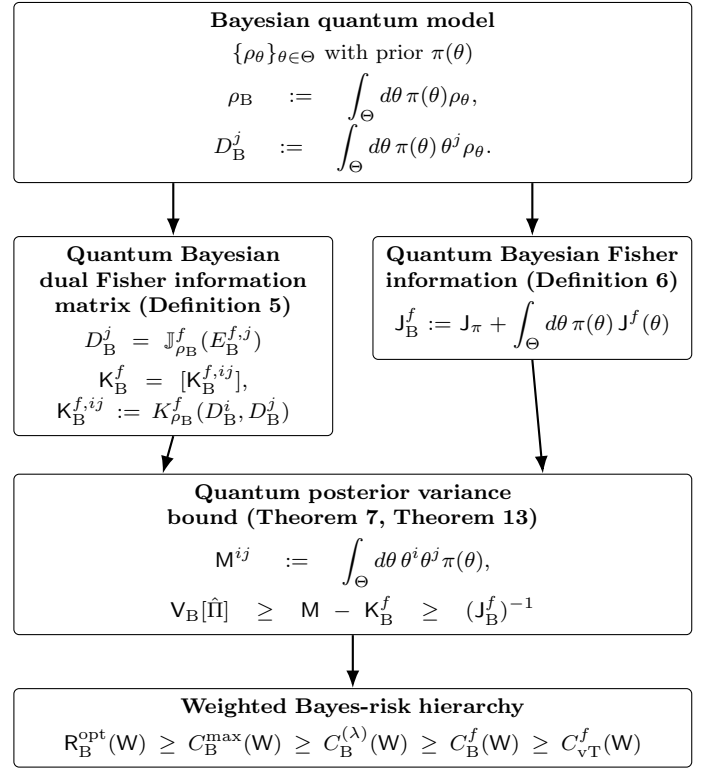
\begin{figure}[t]
\centering
\begin{tikzpicture}[
    >=Latex,
    font=\footnotesize,
    box/.style={
        draw,
        rounded corners=2pt,
        align=center,
        inner sep=3.8pt
    },
    wide/.style={box, text width=8.7cm},
    half/.style={box, text width=3.95cm},
    arrow/.style={->, line width=0.8pt}
]

\node[wide] (top) {%
\textbf{Bayesian quantum model}\\[0.8mm]
$\{\rho_\theta\}_{\theta\in\Theta}$ with prior $\pi(\theta)$\\[0.8mm]
$\displaystyle
\rho_{\mathrm B}:=\int_{\Theta} d\theta\,\pi(\theta)\rho_\theta,$\\
$\displaystyle
D_{\mathrm B}^{j}:=\int_{\Theta} d\theta\,\pi(\theta)\,\theta^{j}\rho_\theta.$
};

\node[half, anchor=north west] (left)
at ($(top.south west)+(0,-7mm)$) {%
\textbf{Quantum Bayesian dual Fisher information matrix (Definition~\ref{def:QPM-QBFI})}\\[0.8mm]
$\displaystyle
D_{\mathrm B}^{j}=\mathbb J_{\rho_{\mathrm B}}^{f}(E_{\mathrm B}^{f,j})$\\[0.8mm]
$\displaystyle
\Ksf_{\mathrm B}^{f}=[\Ksf_{\mathrm B}^{f,ij}],$\\
$ \Ksf_{\mathrm B}^{f,ij}:=
K_{\rho_{\mathrm B}}^{f}(D_{\mathrm B}^{i},D_{\mathrm B}^{j})
$
};

\node[half, anchor=north east] (right)
at ($(top.south east)+(0,-7mm)$) {%
\textbf{Quantum Bayesian Fisher information (Definition~\ref{def:fvT})}\\[0.8mm]
$\displaystyle
\Jsf_{\mathrm B}^{f}
:=\Jsf_{\pi}+\int_{\Theta} d\theta\,\pi(\theta)\,\Jsf^{f}(\theta)$
};

\node[wide, anchor=north] (mid)
at ($(left.south)!0.5!(right.south)+(0,-10mm)$) {%
\textbf{Quantum posterior variance bound (Theorem~\ref{thm:fBLD}, Theorem~\ref{thm:ineq2})}\\[0.8mm]
$\displaystyle
\Msf^{ij}:= \int_\Theta d\theta\, \theta^i\theta^j\pi(\theta),$\\[0.8mm]
$\displaystyle
\Vsf_{\mathrm B}[\hat\Pi]\geq \Msf-\Ksf_{\mathrm B}^{f}
\geq (\Jsf_{\mathrm B}^{f})^{-1}
$
};

\node[wide, anchor=north] (bottom)
at ($(mid.south)+(0,-7mm)$) {%
\textbf{Weighted Bayes-risk hierarchy}\\[0.8mm]
$\displaystyle
\Rsf_{\mathrm B}^{\mathrm{opt}}(\Wsf)\ge
C_{\mathrm{B}}^{\max}(\Wsf)\ge
C_{\mathrm{B}}^{(\lambda)}(\Wsf)\ge
C_{\mathrm{B}}^{f}(\Wsf)\ge
C_{\mathrm{vT}}^{f}(\Wsf)
$
};

\draw[arrow] ($(top.south west)!(left.north)!(top.south east)$) -- (left.north);
\draw[arrow] ($(top.south west)!(right.north)!(top.south east)$) -- (right.north);

\draw[arrow] (left.south) -- ($(mid.north west)!0.22!(mid.north east)$);
\draw[arrow] (right.south) -- ($(mid.north west)!0.78!(mid.north east)$);

\draw[arrow] (mid.south) -- (bottom.north);
\end{tikzpicture}
\caption{
Structural relationships among quantum Bayesian Fisher information,
posterior variance bounds,
and weighted Bayes-risk bounds
in the Bayesian monotone metric framework.
}
\label{fig:bayesian_information_structure}
\end{figure}

The rest of this paper is organized as follows.
In Sec.~\ref{sec:preliminary}, we introduce the basic notation and summarize the key ingredients of Petz monotone metrics that underlie our Bayesian construction.
In Sec.~\ref{sec:Bmonotone}, we define Bayesian monotone metrics on the prior-averaged state and develop the associated canonical Bayesian objects, including quantum posterior-mean operators and the quantum Bayesian dual Fisher-information matrix.
In Sec.~\ref{sec:Bbound}, we derive the resulting family of metric-induced posterior variance lower bounds on the Bayes risk and prove their universal dominance over the corresponding quantum van Trees (quantum Bayesian Cram\'er--Rao) bounds.
In Sec.~\ref{sec:result2}, we establish the reduction from optimization over all operator monotone functions to a one-parameter $\lambda$-subfamily, which turns the tightest quantum posterior variance bound into a tractable single-parameter optimization.
In Sec.~\ref{sec:example}, we present numerical qubit examples demonstrating that our bounds can be strictly tighter than the Bayesian SLD and Bayesian RLD posterior variance bounds, with a clear quantitative gap.
We conclude in Sec.~\ref{sec:conclusion} with a discussion of implications and open directions.
Most technical proofs are deferred to the Appendices.

\section{Preliminaries}\label{sec:preliminary}

\subsection{Notations for the Bayesian MSE}\label{sec:notations}
In this section, we recall the classical Bayesian estimation at first. The notations referring to the Bayesian statistics are listed below: 
\begin{itemize}
    \item $\theta=(\ta^1,\ta^2,\ldots,\ta^n)$: $n$-dimensional parameter. In this paper, we use upper script to denote the component. 
    \item $p_\theta(x)=p(x|\theta)$: A parametric model (likelihood), a family of probability distributions on $\mathcal{X}$.
    \item $\pi(\theta)$: A prior distribution for the parameter space $\Theta$.
    \item $P(\theta,x)$: The joint distribution, $P(\theta,x)=\pi(\theta)p_\theta(x)$.     
    \item $P_X(x)$: The marginal (evidence) for the random variable $X$, which is defined by $P_X(x)=\int_\Theta d\ta\,P(\theta,x)$. 
\end{itemize}
The expectation value of a random variable with respect to the model $p_\theta(x)$, the prior $\pi(\ta)$, the marginal $P_X(x)$, and the joint distribution $P(\theta,x)$ is expressed by $\Ebb_\theta[\cdot]$, $\Ebb_X[\cdot]$, $\Ebb_\pi[\cdot]$, and $\Ebb[\cdot]$, respectively. 

An estimator that returns values on the set $\Theta$ is denoted as 
$\hattheta(x)=(\hattheta^1(x),\hattheta^2(x),\cdots,\hattheta^n(x) )$. 
The main objective of this paper is to minimize the average weighted trace of the MSE matrix.
We consider the most general problem with a weight matrix, which is $n\times n$ positive semidefinite matrix: $\Wsf=[\Wsf_{jk}] \geq 0$. The Bayes risk is then defined by
\begin{align*}
    \Rsf_\mathrm{B} [\hattheta|\Wsf] &= \Tr \left[ \Wsf \Vsf_\Brm [ \hattheta] \right],\\
    \Vsf_\Brm[\hata]&=\left[\Vsf_{\mathrm{B},jk}[\hattheta] \right] ,\\
    \Vsf_{\mathrm{B},jk}[\hattheta]&=  \mathbb E \left[ (\hattheta^j(X)-\theta^j) (\hattheta^k(X)-\theta^k) \right]. 
\end{align*}
Throughout the paper, $\Tr[\cdot ]$ is the trace over $n$-dimensional parameter space. 

\subsection{Optimal estimator and van Trees bound} \label{sec:classicaltheorem}

The optimal estimator for the Bayes risk $\Rsf_\mathrm{B} [\hattheta|\Wsf]$ is proven in Appendix \ref{sec:appBayes}, 
which is given by
\begin{equation} \label{eq:mmse}
\hattheta^j_{\Brm}(x)=\int_\Theta d\ta \theta^j \frac{P(\ta,x)}{P_X(x)}. 
\end{equation}
Let us define the corresponding MSE matrix by $\Vmin:=\VsfB[\hattheta_{\Brm}]$, and this is also expressed as
\begin{equation} \label{eq:CoptBayesRisk}
\Vmin=\Msf-\Ksf_\Brm, 
\end{equation}
where $\Msf=[\Msf^{ij}]$ and $\Ksf_\Brm=[\Ksf^{ij}_{\Brm}]$ are
\begin{align}
\Msf^{ij}&:= \int_\Theta d\ta\, \ta^i\ta^j\pi(\ta),\\
\Ksf^{ij}_{\Brm}&:=\int_\Xcal dx\, \frac{d^i_{\Brm}(x)d^j_{\Brm}(x)}{P_X(x)} \label{def:CKmatrix},\\
d^i_{\Brm}(x)&:= \int_\Theta d\ta \theta^i P(\ta,x).
\end{align}
Note that the optimal estimator Eq.~\eqref{eq:mmse} satisfies the relation:
\begin{equation} \label{eq:cLD}
d^i_{\Brm}(x)= \hattheta^i_{\Brm}(x) P_X(x), 
\end{equation}
which plays an important role when discussing the quantum case. 

The van Trees lower bound for the Bayes risk is defined by the Bayesian Fisher information, 
\begin{align}
\cvT(\Wsf)&:= \Tr\left[\Wsf\Jsf_\Brm^{-1}\right],\\
\Jsf_{\mathrm{B},ij}&:= \Ebb\left[ \frac{\del\log P(\theta,X)}{\del \ta^i} \frac{\del\log P(\theta,X)}{\del \ta^j}\right]. 
\end{align}
The matrix inequality $\Vsf_\Brm[\hata] \geq \Jsf_{\mathrm{B}}^{-1}$ is also known as the Bayesian Cram\'er-Rao inequality. 

The following statement is known in classical statistics: 
\begin{theorem} \label{thm:Cmain}
For any prior under a weak boundary condition, the matrix inequality holds. 
\[
\Vmin\geq \Jsf_\Brm^{-1}. 
\]
\end{theorem}
Here the weak boundary condition is to require that $\ta^i\pi(\ta)$ vanishes at the boundary $\Theta$. 
In Appendix~\ref{sec:appBayes}, we provide an alternative proof of this theorem, which will be useful for extending the argument to the quantum Bayesian setting.

\subsection{Quantum Bayesian estimation}\label{sec:quantumBE}
We aim at deriving lower bounds in another direction. Before that, we need to define the problem in a quantum Bayesian setting. Let $\Hil$ be a finite dimensional Hilbert space and $\left\{ \rho_\theta~|~ \theta \in \Theta \right\}$ be an $n$-parameter quantum model with $\theta=(\theta^1, \theta^2, \cdots, \theta^n)$. Unlike point estimation, we do not need to impose many regularity conditions such as smoothness of the parameter dependence, full-rankness for the model $\rho_\theta\geq0$, etc. Let a set of positive semidefinite matrices $\{\Pi_x\}$ be a measurement with the measurement outcome $x\in\Xcal$. The quantum measurement is typically referred to as a positive operator-valued measure (POVM), which is defined by
\begin{equation*}
    \Pi=\{ \Pi_x \},~ \forall x \in \Xcal,~ \Pi_x \geq 0,~ \sum_{x\in\Xcal} \Pi_x = I,
\end{equation*}
where $I$ is the identity operator on $\Hil$. 
When the measurement is labelled by continuous numbers, we use the integration instead. 

Measurement outcome is described by a random variable $X$ that obeys the conditional probability distribution (a model or likelihood): 
\[
p_\theta(x) = \tr \left[ \rho_\theta \Pi_x \right],
\] 
where $\tr[ \cdot]$ denotes the trace on the Hilbert space $\Hil$. 
The performance of the estimator in this study is quantified by the quantum Bayes risk:
\begin{equation*}\label{def:Brisk2}
  \Rsf_\mathrm{B} [\Pi,\hata|\Wsf]= \Tr\left[ \Wsf \Vsf_\Brm[\Pi,\hata]\right].
\end{equation*}
The main objective is to find the best pair $\hPi=(\Pi, \hattheta)$ that minimizes the Bayes risk.
\begin{equation}
\Rsf^{\text{opt}}_\Brm(\Wsf):=\min_{\Pi,\hata}\Rsf_\mathrm{B} [\Pi,\hata|\Wsf]. 
\end{equation}

\subsection{Bayesian Nagaoka--Hayashi bound}\label{sec:BNHbound}

The Bayesian Nagaoka--Hayashi (BNH) bound, introduced in Ref.~\cite{suzuki2024bayesian}, provides an SDP lower bound on the Bayes risk. 
For a given weight matrix \(\Wsf\), we denote this bound by \(\cNH(\Wsf)\). 
It satisfies
\begin{align}
    \Rsf^{\mathrm{opt}}_\Brm(\Wsf)\geq \cNH(\Wsf),
\end{align}
where \(\Rsf^{\mathrm{opt}}_\Brm(\Wsf)\) is the optimal Bayes risk defined above.

The BNH bound is an important benchmark for Bayesian quantum estimation. 
However, its evaluation requires solving an optimization problem over an operator-valued matrix and an operator-valued vector. 
Since the BNH bound is not the central object of the present paper, we defer its precise formulation and derivation to Appendix~\ref{App:BNH}. 
In the main text, we focus instead on closed form lower bounds induced by Bayesian monotone metrics.

\section{Bayesian lower bounds}\label{sec:result}
In this section, we shall propose a method of deriving two-different types of families of lower bounds based on an operator monotone metric. More detailed discussion is given in Appendix \ref{sec:AppOpMono}. 

\subsection{Bayesian monotone metric} \label{sec:Bmonotone}
A given operator monotone function $f:(0,\infty)\to(0,\infty)$, we define an inner product with respect to a state $\rho>0$ on the matrix space as follows. Note that we demand the condition $f(1)=1$ in order to recover the classical result. 
\begin{definition}[Monotone metric]
\[
K^f_\rho(X,Y):= \tr[X^\dagger (\Jbb^f_\rho)^{-1}Y]\mbox{ for }X,Y\in\Cbb^{d\times d}. 
\]  
Here $\Jbb^f_\rho$ is a super-operator acting on $X\in\Cbb^{d\times d}$ by 
\[
\Jbb^f_\rho (X)= f(\Delta_\rho)X \rho, 
\]
with $\Delta_\rho(X)=\rho X\rho^{-1}$ the modular operator. 
\end{definition}
In this definition, $\Jbb^f_\rho$ is a positive operator and hence its inverse exists. 
A brief account on operator monotone metrics and operator means is given in Appendix \ref{sec:AppOpMono}. 

Let $\Ecal$ be a completely-positive and trace-preserving (CP-TP) map, a metric is said monotone, if 
\begin{equation} \label{eq:monotone}
K^f_\rho(X,X) \geq K^f_{\Ecal(\rho)}(\Ecal(X),\Ecal(X)), 
\end{equation}
holds for any CP-TP map (see Appendix \ref{sec:AppOpMono} for more precise definition.). 
The Petz theorem states that there exists one-to-one corresponding to a monotone metric and an operator monotone function $f$. 
In the standard setting of parameter estimation of quantum states $\rho_\theta$, a quantum Fisher information matrix is defined by each $f$. 
\begin{definition}
For a given $f$, we define $f$LD operators $L^f_{\ta,j}$ ($j=1,2,\ldots,n$) by the solution to 
\begin{equation} \label{eq:fLD_op}
\del_j \rho_\ta = \Jbb^f_{\rho_\ta} (L^f_{\ta,j}),  
\end{equation}
where $\del_j:=\frac{\del }{\del \theta^j} $ denotes the partial derivative. 
A quantum Fisher information matrix $\Jsf^f(\theta)$ associated with $f$ is defined by the inner product as 
\[
\Jsf^f_{ij}(\theta):=K^f_\rho(\del_i \rho_\ta,\del_j \rho_\ta). 
\]
\end{definition}
This is also expressed in the equivalent form as follows. (see Appendix \ref{sec:AppOpMono}.) 
Define another inner product
\begin{equation} \label{def:Erep_inner}
\inner{X}{Y}^f_{\rho}:= \tr[X^\dagger \Jbb^f_\rho Y],
\end{equation}
then $\Jsf^f_{ij}(\theta)$ is 
\[
\Jsf^f_{ij}(\theta)=\inner{L^f_{\ta,i}}{L^f_{\ta,j}}^f_{\rho_\theta}. 
\] 

Familiar examples of quantum Fisher information matrices are defined by 
\[
f_{\text{S}}(t)=\frac12(1+t),\ f_{\text{R}}(t)=t.  
\]
$f_{\text{S}}$ and $f_{\text{R}}$ correspond to the SLD and RLD Fisher information, respectively. 
An important class of operator monotone functions was pointed out by one of authors \cite{yamagata2021maximum}. 
\begin{equation}
f_\la(t):=\mlam+\plam t\quad(\la\in[-1,1]), 
\end{equation}
which will be called the $\lambda$LD function. The corresponding LD and Fisher information are denoted by
\[
L^{(\la)}_{\ta,j}\ \mbox{and}\ \Jsf^{(\la)}_{ij}(\theta). 
\]

A Bayesian version of logarithmic derivative-type equation is introduced as follows. 
This naturally invites us to define the Bayesian monotone metric.
\begin{definition}[Bayesian monotone metric]\label{def:Bmonotone}
Let $\rho_\Brm$ be the averaged state under the prior $\pi$.
For operators $X,Y$, define
\[
K^f_{\rho_\Brm}(X,Y)
:= \tr\!\left[X\,\big(\Jbb^f_{\rho_\Brm}\big)^{-1}(Y)\right].
\]
We call $K^f_{\rho_\Brm}(\cdot,\cdot)$ the Bayesian monotone metric associated with
the operator monotone function $f$.
\end{definition}

\begin{definition}[Quantum posterior-mean operator and Quantum Bayesian dual Fisher information matrix]\label{def:QPM-QBFI}
Given $D^j_\Brm$, the solution $E^{f,j}_\Brm$ of
\begin{equation}\label{eq:fLD_bayes}
D^j_{\Brm} = \Jbb^f_{\rho_\Brm}\!\left(E^{f,j}_{\Brm}\right)
\end{equation}
is called the quantum posterior-mean operator. With respect to the Bayesian monotone metric \(K^f_{\rho_\Brm}\), define the Quantum Bayesian dual Fisher information matrix $\Ksf^f_\Brm=[\Ksf^{f,ij}_\Brm]$ by
\[
\Ksf^{f,ij}_\Brm := K^f_{\rho_\Brm}\!\left(D^i_{\Brm},D^j_{\Brm}\right).
\]
\end{definition}
Remarks: 
\begin{itemize}
  \item First, note that all the quantities are defined through the expectation values over the prior $\pi$.
  \item Second, the quantum posterior-mean operator \(E^{f,j}_{\Brm}\) can be regarded as a quantum version of the minimum mean-square estimator \eqref{eq:mmse}. In particular, the relation \eqref{eq:cLD} is the corresponding classical equation. In contrast to the classical case, the equation
\(
D = \Jbb^{f}_{\rho_\Brm}(E)
\)
is solved through a linear map on operators, rather than by pointwise division.
If the averaged state satisfies \(\rho_\Brm>0\), then \(\Jbb^{f}_{\rho_\Brm}\) is invertible and the solution is unique:
\begin{align}
E = \big(\Jbb^{f}_{\rho_\Brm}\big)^{-1}(D).
\end{align}
If \(\rho_\Brm\) is rank-deficient, \(\Jbb^{f}_{\rho_\Brm}\) may have a non-trivial kernel, and the operator solution \(E\) need not be unique unless one restricts to the support of \(\rho_\Brm\) or uses a generalized inverse. 
However, this non-uniqueness does not affect the induced metric quantities: the Bayesian monotone metric and the matrix \(\Ksf^f_\Brm\) are uniquely determined on the support of \(\rho_\Brm\).

  \item Third, the $f$-posterior-mean equation \eqref{eq:fLD_bayes} shares the exactly same algebraic structure as point estimation \eqref{eq:fLD_op}. 
The correspondence is 
\begin{align}
D^j_{\Brm}\leftrightarrow \del_j \rho_\ta \\ 
E^{f,j}_{\Brm}\leftrightarrow L^{f}_{\ta,j}.   
\end{align}
\item Fourth, using another inner product representation~\eqref{def:Erep_inner}, we have
\[
\Ksf^{f,ij}_\Brm= \inner{E^{f,i}_{\Brm}}{E^{f,j}_{\Brm}}^f_{\rho_\Brm}.
\]
  \item Last, when we consider the $\lambda$-posterior-mean equation \eqref{eq:lambdaposteriormeaneq}, we shall denote the above quantities as 
$E^{(\la),j}_{\Brm}$ and $\Ksf^{(\la)}_\Brm$. 
\end{itemize}  
Another natural candidate for a Bayesian version of monotone metric is introduced by the same idea of van Trees bound. 
\begin{definition}[Quantum Bayesian Fisher information] \label{def:fvT}
Given a model $\rho_\ta$ and a prior $\pi$, we define 
the averaged $f$LD Fisher information matrix over the prior as 
\[
\Jsf^f_\Brm:= \Jsf_\pi + \int_\Theta d\ta\,\pi(\theta)\,  \Jsf^f(\theta), 
\]
where $\Jsf_{\pi,ij}:=\Ebb_\pi[\frac{\del\log \pi(\ta)}{\del \ta^i}\frac{\del\log \pi(\ta)}{\del \ta^j}]$ is the classical Fisher information about the prior. 
\end{definition}

Historical remarks are in order. 
Personick introduced two lower bounds based on the SLD for one-parameter models \cite{personick1971application}. Rubio and Dunningham generalized the Personick bound in multiparameters \cite{rubio2020bayesian}. They showed that $\Vsf_\Brm[\hPi]\geq\Msf-\Ksf^{(\la=0)}_\Brm \geq (\Jsf^{(\la=0)}_\Brm)^{-1}$ holds. 
Later, Holevo introduced a lower bound based on the RLD, $\Msf-\Ksf^{(\la=1)}_\Brm$ \cite{holevo1973statistical}. 
A related derivation can be found in Wang and Hayashi {\it et al.} \cite{wang2007quantum}. 
After renewed interest in multiparameter estimation, several authors generalized 
the Personick bounds to multiparameter estimation. 
In the next section, we generalize all these results based on the formalism of operator monotone functions. 

\subsection{Bayesian monotone metric bounds} \label{sec:Bbound}
This section provides the main result of this paper, the Bayesian monotone metric bounds.

\subsubsection{Quantum posterior variance bound}
The quantum posterior variance bound is given as follows. 
\begin{theorem}[Quantum posterior variance]\label{thm:fBLD}
Given any operator monotone function $f$, the Bayesian MSE is bounded by Quantum posterior variance
\[
\Vsf_\Brm[\hPi]\geq\Msf-\Ksf^f_\Brm. 
\]
\end{theorem}
By using the well-known lemma, we can derive a lower bound for 
the Bayes risk. (see Appendix \ref{sec:App_proof1} for the proof.)
\begin{corollary}[Quantum posterior variance bound] \label{cor:fBLD}
\begin{align*}
\Tr \left[\Wsf \Vsf_\Brm[\hPi]\right]&\geq \cfBLD(\Wsf),\\
\cfBLD(\Wsf)&:=\Tr\left[\Wsf \Msf\right]-\Tr\left[\Wsf \Real\Ksf^f_\Brm\right]\\
&\qquad+\Tr\left|\Wsf^{\frac12} \Imag\Ksf^f_\Brm\Wsf^{\frac12} \right| , 
\end{align*}
where $\Tr|X|:=\Tr[\sqrt{X^\dagger X}]$ denotes the trace norm. 
\end{corollary}

\begin{proof}[Proof: Theorem \ref{thm:fBLD}]\ 
The statement is proven by showing the inequality: 
For any choice of operator monotone functions, the matrix inequality 
\begin{equation}\label{eq:cqBayes}
\Ksf_\Brm[\Pi]\leq \Ksf^f_\Brm,
\end{equation}
holds for any POVM. 
The left hand side of this inequality, $\Ksf_{\Brm}[\Pi]$, is the matrix \eqref{def:CKmatrix}, which is defined by
\[
\Ksf_{\Brm}^{ij}[\Pi]:= \sum_{x\in\Xcal}\, \frac{\tr[D^i_{\Brm} \Pi_x]\tr[D^j_{\Brm} \Pi_x]}{\tr[\rho_\Brm \Pi_x]}. 
\]
The proof for Eq.~\eqref{eq:cqBayes} is given in Appendix \ref{sec:App_proof_cq}. 
This is because applying the optimal estimator $\hattheta_{\Brm}$ Eq.~\eqref{eq:mmse} yields the classically optimal Bayes risk \eqref{eq:CoptBayesRisk} as
\[
\Vsf_\Brm[\hPi]\geq \Msf-\Ksf_{\Brm}[\Pi]\geq \Msf-\Ksf^f_\Brm. 
\]
\end{proof}

\subsubsection{Quantum van Trees bound}
The quantum van Trees bound is easily generalized to any monotone metric as follows. 
\begin{theorem}[Quantum Bayesian CR bound] \label{thm:fvT}
For any operator monotone function $f$, define the averaged Fisher information $\Jsf^f_\Brm$ by Definition \ref{def:fvT}, 
then the Bayes risk is bounded as
\[
\Vsf_\Brm[\hPi]\geq (\Jsf^f_\Brm)^{-1}. 
\]
\end{theorem}
By the same argument to get Corollary \ref{cor:fBLD}, we have a lower bound for the weighted trace of the Bayesian MSE matrix. 
\begin{corollary} \label{cor:fvT}
\begin{align*}
\Tr \left[\Wsf \Vsf_\Brm[\hPi]\right]&\geq \cfvT(\Wsf),\\
\cfvT(\Wsf)&:=\Tr\left[\Wsf \Real(\Jsf^f_\Brm)^{-1}\right]
+\Tr\left|\Wsf^{\frac12} \Imag (\Jsf^f_\Brm)^{-1}\Wsf^{\frac12} \right|. 
\end{align*}
\end{corollary}

While it is possible to prove Theorem \ref{thm:fvT} directly, we shall prove it by showing a quantum version of Theorem \ref{thm:Cmain} in the next section (Theorem \ref{thm:ineq2}). 

\subsection{Inequalities} \label{sec:inequality}
In this section, we prove the second main result of this paper. 
The first inequality states that the $\lambda$LD function provides the best choice among all possible monotone functions. 
This fact was used by one of authors in Ref.~\cite{yamagata2021maximum} in point estimation. 
We can generalize it to the Bayesian setting. 
\begin{theorem} \label{thm:ineq1}
For any operator monotone function $f$, there always exists $\la\in[-1,1]$ such that 
\[
\Ksf^{(\la)}_\Brm \leq \Ksf^f_\Brm. 
\]
\end{theorem}
The proof is given in Appendix \ref{sec:App_proof2}. 
This theorem shows that it is sufficient to consider the $\lambda$-posterior mean equation as a sub-family of operator monotone functions as far as lower bounds are concerned. 
In other words, we have the following corollary. 
\begin{corollary} \label{cor:ineq1}
Consider all possible posterior variance bounds $\cfBLD(\Wsf)$ associated with $f$, it is sufficient to maximize the $\lambda$-posterior mean family. 
\begin{align}
\max_{f:\mathrm{operator\ monotone}}\cfBLD(\Wsf)=\max_{\la\in[-1,1]}\Ccal^{f_\la}_{\mathrm{B}}(\Wsf).   
\end{align}
\end{corollary}
The precise statement will be given in Theorem \ref{thm:ineq3} later. 

Remark: The exactly same argument of Ref.~\cite{yamagata2021maximum} shows the hierarchy among the $\Jsf^f_\Brm$s. 
For any operator monotone function $f$, there always exists $\la\in[-1,1]$ such that 
\[
\Jsf^{(\la)}_\Brm \leq \Jsf^f_\Brm. 
\]
In other words, the family $\{\Jsf^{(\la)}_\Brm\}_{\la \in[-1,1]}$ based on the $\lambda$-posterior mean dominates all others. 

The second main result is a quantum counterpart of Theorem \ref{thm:Cmain}. 
\begin{theorem} \label{thm:ineq2}
Suppose the weak boundary condition is satisfied for the quantum posterior-mean operator.
For any operator monotone function $f$, the following matrix inequality holds. 
\begin{align}
\Msf-\Ksf^f_\Brm \geq (\Jsf^{f}_\Brm)^{-1}.   
\end{align}
\end{theorem}

Remarks: This theorem shows that all possible quantum van Trees bounds do not provide better lower bounds, since there always exists a better matrix $\Ksf^f_\Brm$ associated with the same $f$ as a matrix inequality. 
In the next section, we will show that this relation proves explicit ordering for the weighted-trace versions in Theorem \ref{thm:ineq3}. 

\begin{proof}[Proof: Theorem \ref{thm:ineq2}]
To prove this theorem, we need three steps. Some of technical parts will be deferred to the Appendix. 
The first step is to rewrite the left hand $\Msf-\Ksf^f_\Brm$ in terms of a single matrix defined through an inner product between the set of operators. Due to non-commutativity, this is not possible in general. 
Instead, we use the concavity fact: for any fixed operator $X$,
the map
\begin{align}
\rho \ \longmapsto\ \langle X, X\rangle_\rho^f
= \tr\!\left[X^\dagger \mathbb J_\rho^f(X)\right]
\label{eq:concavity_fact}
\end{align}
is concave in $\rho$. A detailed proof is given in Appendix~\ref{app:concavity}.

Thus, for any $c\in\mathbb{C}^n$, define $E_c^{f}:=\sum_i c_i E_\Brm^{f,i}$.
This concavity implies the inequality
\begin{align}\label{eq:concavity_core}
\int_\Theta d\theta\,\pi(\theta)\,\langle E_c^{f},E_c^{f}\rangle_{\rho_\theta}^f
\le \langle E_c^{f},E_c^{f}\rangle_{\rho_\Brm}^f
= c^\dagger \Ksf_\Brm^f c,  
\end{align}
which proves the matrix inequality
\begin{align}
\Ksf_\Brm^f \ \geq\ \Big[\,\langle E_\Brm^{f,i}, E_\Brm^{f,j}\rangle_{\pi(\theta),\rho_\theta}^f \,\Big].
\end{align}
Here the averaged inner product is defined by
\begin{align}
\langle X, Y\rangle_{\pi(\theta),\rho_\theta}^f
:= \int_\Theta d\theta\,\pi(\theta)\,\langle X, Y\rangle_{\rho_\theta}^f .
\end{align}
With an additional argument proven in Appendix \ref{sec:App_step1}, we derive the following inequality. 
\begin{align} \label{eq:step1}
\Msf-\Ksf^f_\Brm &\geq\Vsf^{f}_\Brm,\\
\Vsf^{f}_{\Brm,ij}&:=
\inner{E^{f,i}_{\Brm}-\theta^i I}{E^{f,j}_{\Brm}-\theta^j I}^{f}_{\pi(\ta),\rho_\ta}. \nonumber
\end{align}

The second step is a quantum version of Lemma \ref{lem:Bscore}, which is proven in Appendix \ref{sec:App_step2}. 
With this lemma \ref{lem:QBscore}, we show orthogonality between the two sets of operators with respect to $\inner{\cdot}{\cdot}^{f}_{\pi(\theta), \rho_\theta}$. 
\begin{align}
\mbox{set 1:}&\quad\left\{E^{f,i}_{\Brm}-\theta^i I \right\}_{i=1}^n,\\ 
\mbox{set 2:}&\quad\left\{ L^{f}_{\ta,j}+\frac{\del\log \pi(\ta)}{\del \ta^j} I\right\}_{j=1}^n.  
\end{align}
That is setting $X_\ta=E^{f,i}_{\Brm}-\theta^i I$ yields (see also Appendix \ref{sec:App_step3})
\begin{align*}
\left\langle E^{f,i}_{\Brm}-\theta^i I\,,\,L^{f}_{\ta,j}+\frac{\del\log \pi(\ta)}{\del \ta^j} I\right\rangle^{f}_{\pi(\ta),\rho_\ta} = \,
&\delta^i_j. 
\end{align*}
The last step is exactly same as the classical one. 
We consider an $2n\times2n$ matrix such that 
the $ij$ block matrices for $i,j=1,2$ are defined by the inner product $\inner{\cdot}{\cdot}^{f}_{\pi(\ta),\rho_\ta}$ between set $i$ and set $j$. 
This gives a positive semi-definite matrix (see Appendix \ref{sec:App_step3}),
\begin{equation} \label{eq:step3}
\left[\begin{array}{cc} \Vsf^{f}_\Brm & \Isf \\
 \Isf & \Jsf^{f}_\Brm \end{array}\right]\geq 0 
\iff \Vsf^{f}_\Brm\geq(\Jsf^{f}_\Brm)^{-1}. 
\end{equation}
Combining Eqs.~(\ref{eq:step1}, \ref{eq:step3}) proves the claimed inequality. 
\[
\Msf-\Ksf^f_\Brm \geq\Vsf^{f}_\Brm\geq (\Jsf^{f}_\Brm)^{-1}. 
\] 
\end{proof}

\section{One-parameter family of lower bounds} \label{sec:result2}

\subsection{Bayesian $\lambda$-posterior variance bound}\label{sec:BlaLD}
The goal in this section is to derive a family of quantum Bayesian lower bounds in a closed based on Theorem \ref{thm:ineq1}. 
The key ingredient in our approach is to utilize the corresponding operator equation which 
was shown to be useful for point estimation \cite{yamagata2021maximum,suzuki2021non}. 
For convenience, we call this family the Bayesian $\lambda$-posterior variance bound. 

\begin{definition}[Bayesian $\lambda$-posterior variance bound] \label{def:BlaLD}
For any $\lambda\in[-1,1]$, we define the Bayesian $\lambda$-posterior variance bound by
\begin{align*}
    \Ccal^{(\la)}_\mathrm{B}(\Wsf)&:=\Tr[\Wsf\Msf]  -\Tr \left[\Wsf \Real \Ksf^{(\la)}_\Brm\right] \\
   &\hspace{2cm} + \Tr | \Wsf^{\frac12} \Imag \Ksf^{(\la)}_\Brm \Wsf^{\frac12} | .
\end{align*}
In this bound, a complex positive definite matrix $\Ksf^{(\la)}_\Brm$ is defined by
\begin{equation}\label{def:Kmatrix}
    \Ksf^{(\la)}_\Brm:= \left[ \inner{E_{\Brm}^{(\la),j}}{E_{\Brm}^{(\la),k}}_{\rho_\Brm}^{(\la)} \right],
\end{equation}
where $E_{\Brm}^{(\la),j}$ is the solution to the quantum posterior-mean equation:
\begin{equation} \label{eq:lambdaposteriormeaneq}
    D^j_{\Brm} = \plam \rho_{\Brm} E_{\Brm}^{(\la),j} + \mlam E_{\Brm}^{(\la),j} \rho_{\Brm}. 
\end{equation}
\end{definition}

We remark that due to the process of defining the $\lambda$-posterior variance bound, it is tighter than the Personick bound \cite{personick1971application} and its generalization \cite{rubio2019quantum,demkowicz2020multi}, which is equal to the $\lambda=0$ case. 
On the other hand, the Bayesian bound proposed by Holevo \cite{holevo_qest,holevo1977commutation} corresponds to the choice $\lambda=1$. In other words, the above one-parameter family of bounds includes well-known bounds as special cases.  

The optimal choice of \(\lambda\) in the Bayesian \(\lambda\)-posterior variance bound, \(\Ccal^{(\la)}_\mathrm{B}(\Wsf)\), is highly non-trivial. 
To obtain analytical insight into this optimization, we solve the general qubit model explicitly in Sec.~\ref{sec:generalqubit}.

\subsection{Maximum Bayesian posterior variance bound}\label{sec:maxBLD}
In point estimation, one of authors proposed the maximum logarithmic derivative (MLD) bound by considering a one-parameter family of quantum LD bounds \cite{yamagata2021maximum}. It is then natural to consider a Bayesian version of the MLD bound.  
A Bayesian version of the MLD bound is presented in the following corollary.
\begin{corollary}\label{cor:maxLD}
    The Bayes risk is bounded by the maximum posterior variance bound
\begin{equation*}
    \Rsf_{\mathrm{B}}[\hPi|\Wsf] \geq \Ccal^{\max}_\mathrm{B}(\Wsf) :=\max_\la 
    \Ccal^{(\la)}_\mathrm{B}(\Wsf).
\end{equation*}
\end{corollary}
We call this bound the maximum posterior variance bound which is the unification of the bounds proposed by Personick and Holevo. Due to the maximization over choice of $\lambda$, this bound is larger than the previously known two bounds. 

The relationship between the $\lambda$-posterior variance bound and the maximum posterior variance bound is elaborated upon in the subsequent sections. The $\lambda$-posterior variance bound is one-parameter family and each of them is in a closed form which means it can be obtained through a straightforward procedure. The maximum posterior variance bound is obtained by maximizing over this one-parameter family, ensuring it consistently outperforms bounds derived from any specific choice of $\lambda$. 

In passing, we also mention that the posterior variance bound is not better than the BNH bound. 
However, unlike the BNH bound, which is expressed as optimization, the posterior variance bound is algebraically solved. 

Finally, combining all the inequalities in the previous section and here, we show that the maximum posterior variance bound is the best lower bound among two-different families. 
\begin{theorem} \label{thm:ineq3}
The following inequality holds for any monotone function $f$. 
\[
\Rsf_\Brm^{\mathrm{opt}}(\Wsf)\geq \cNH(\Wsf)\geq\Ccal^{\max}_\mathrm{B}(\Wsf)\geq \cfBLD(\Wsf)\geq \cfvT(\Wsf). 
\]
\end{theorem}
\begin{proof}[Proof: Theorem \ref{thm:ineq3}]
The first inequality is shown in Theorem \ref{thm:BNHbound}. 
The second inequality is proven in Appendix \ref{sec:App_BNH-BLD}. 
The third inequality is due to Corollary \ref{cor:maxLD} and Corollary \ref{cor:ineq1}. 
The fourth inequality is proven from Theorem \ref{thm:ineq2}. 
\end{proof}

Before we moving to the next subsection, we discuss the advantages of the proposed bound. 
Unlike point estimation, the model $\rho_\ta$ does not require the differentiability with respect to the parameter. 
The matrix $\Ksf_\Brm^{(\la)}$ is indeed calculated without taking the partial derivatives. 
The quantum version of the van Trees inequality, on the other hand, the model needs to be smooth \cite{personick1970efficient,personick1971application}. 
The posterior-mean operators defined by Eq.~\eqref{eq:fLD_bayes} exist uniquely even for rank-deficient models. 
This is because the averaged state $\rho_\mathrm{B}$ becomes full rank after the integration over the parameter in the following sense. 
Although each $\rho_\theta$ may be rank-deficient,
the averaged state
\begin{align}
\rho_\Brm=\int_\Theta d\theta\,\pi(\theta)\rho_\theta  
\end{align}
is strictly positive on its support.
Therefore, without loss of generality, we restrict the Hilbert space to the support of \(\rho_\Brm\), where \(\rho_\Brm>0\) holds. On this restricted space, the posterior-mean operators are uniquely determined. If one works outside this support, the operator solution may not be unique, but the induced Bayesian monotone metric and the matrix \(\Ksf^f_\Brm\) are uniquely defined on the support.
Bayesian averaging can effectively enlarge the support
when the supports of $\rho_\theta$ vary with $\theta$.
Even if each $\rho_\theta$ is rank-deficient,
the averaged state may become full rank on the restricted space
provided the states do not share a common null vector
on the support of the prior distribution. 
This reduction simplifies the numerical computation
and guarantees uniqueness of the solution.

The non-triviality of the maximum posterior variance bound comes from the fact that there is no ordering 
between the matrices $\Ksf_\Brm^{(\la)}$ for $\la\in[-1,1]$ and the role of the imaginary part should be stressed. 
If we utilize the real part of $\Ksf_\Brm^{(\la)}$ only, we have a rather simple result. 

Lastly, the proposed bound can be computed for a multiple copy setting. 
Suppose we wish to estimate the parameter by performing a collective measurement on $N$-copies of the state $\rho_\ta^{\otimes N}$ ($N>1$). 
The bound $\Ccal^{\max}_\mathrm{B}$ is still valid for this case, and it is given by calculating the matrix $\Ksf_\Brm^{(\la)}$ 
for the state $\rho_\ta^{\otimes N}$. 
An important difference from point estimation is non additivity of the bound. 
In point estimation, the quantum Cram\'er-Rao bound is additive and hence, $N$-copies case is simply obtained by 
a scaling factor $1/N$. In contrast, the maximum Bayesian posterior variance bound is not additive. 
This is seen by the fact that averaged state $\int_\Theta d\ta \pi(\ta)\rho_\ta^{\otimes N}$ is no longer a product state. 
Then, Bayesian posterior variance bound cannot be expressed in terms of $N=1$ case. Hence, $\Ksf_\Brm^{(\la)}$ is not given by a simple multiple factor for the $N=1$ case.

\subsection{General qubit model}\label{sec:generalqubit}
In this section, we give the formula of the Bayesian $\lambda$-posterior variance bound for the general qubit model. Suppose we parameterize a qubit model expanded in terms of the Pauli matrices by 
\begin{align*}
    \rho_\ta = \half ( I+s_\ta \cdot \sigma) ~ \text{with }s_\ta =[s^i_{\ta}] .
\end{align*}
Since single parameter estimation is solved by Personick, we focus on multiple parameter estimation for $n=2,3$.
We denote the expectation of the Bloch vector and the parameter by
\begin{align*}
    s^i_{\mathrm{B}} &:= \int_\Theta d\theta\pi(\theta){s^i_{\ta}}\quad (i=1,2,3) ,\\ 
    \mu^j&:=\int_\Theta d\theta\pi(\theta) \theta^j \quad(j=1,\ldots,n).
\end{align*}
Let $\Csf_\Brm$ be an $n\times 3 $ real matrix defined by 
\begin{align*}
   \Csf_\mathrm{B}=[\Csf_{\mathrm{B}}^{ij}],\, \Csf_{\mathrm{B}}^{ij} := \int_\Theta d\theta\pi(\theta)\ta^i s_{\ta}^{j}-\mu^is_{\mathrm{B}}^{j},
\end{align*}
and $\Fsf_\Brm$ is a $3 \times 3$ anti-symmetric matrix defined by
\begin{equation}\label{def:Fmatrix}
    \Fsf_\Brm = \begin{pmatrix}
        0 & -s_{\mathrm{B}}^{3} & s_{\mathrm{B}}^{2} \\
        s_{\mathrm{B}}^{3} & 0 & -s_{\mathrm{B}}^{1} \\
        -s_{\mathrm{B}}^{2} & s_{\mathrm{B}}^{1} & 0 
    \end{pmatrix}.  
\end{equation}
In the following, we also use the ket-bra notation for vectors as
\begin{align*}
\ket{s_{\mathrm{B}}}&=[s_{\Brm}^{i}] \in \R^3 ,\\
\ket{\mu}&=[\mu^j]\in\R^n.
\end{align*}
The covariance matrix for the parameter is defined by
\begin{equation}
\Csf_\pi=\left[\int_\Theta d\theta\pi(\theta)(\ta^i-\mu^i)(\ta^j-\mu^j) \right].
\end{equation}

\subsubsection{Analytical expression of the matrix $\Ksf_\Brm^{(\la)}$ Eq.~\eqref{def:Kmatrix}}
Firstly, we calculate the formula of $\Ksf_\Brm^{(\la)}$ which is defined by Eq.~\eqref{def:Kmatrix}. 
The procedure to obtain $\Ksf_\Brm^{(\la)}$ is similar to point estimation, and the detail calculation is omitted. 
We first solve the operators $E_{\Brm}^{(\la),j}$, and then to calculate the inner product. 
This gives the form of $\Ksf_\Brm^{(\la)}$ as
\begin{align} 
\Ksf_\Brm^{(\la)} &= \ketbra{\mu}{\mu} + \Csf_\mathrm{B}\nonumber
     \left[ I - \ketbra{s_\mathrm{B}}{s_\mathrm{B}} +i \la \Fsf_\mathrm{B} \right]^\inv \Csf_\mathrm{B}^\top\\
     &=\ketbra{\mu}{\mu}+  \frac{1}{1-\la^2 |s_\mathrm{B}|^2} \nonumber \\
     &\hspace{5mm}\times\Csf_\mathrm{B}\Big[I+\frac{1-\la^2}{1-|s_\mathrm{B}|^2}\ketbra{s_\mathrm{B}}{s_\mathrm{B}} -i\lambda  \Fsf_\Brm \Big] \Csf_\mathrm{B}^\top .\label{eq:qubitK}
\end{align}

\subsubsection{Explicit form of the Bayesian $\lambda$-posterior variance bound $\Ccal_{\mathrm{B}}^{(\la)}$}
By definition, we get the formula of the Bayesian $\lambda$-posterior variance bound for the general qubit model. 
\begin{theorem}\label{thm:Clam}
The Bayesian $\lambda$-posterior variance bound for the general qubit model is given by
\begin{multline*}
   \Ccal_\mathrm{B}^{(\la)}(\Wsf) = \Tr [\Wsf\Csf_\pi]- 
\frac{1}{1-\la^2 |s_\mathrm{B}|^2}\Big[ \Tr{[\Csf_\mathrm{B}^\top\Wsf\Csf_\mathrm{B}]}\\
+\frac{1-\la^2}{1-|s_\mathrm{B}|^2}\bra{s_\mathrm{B}} \Csf_\mathrm{B}^\top\Wsf \Csf_\mathrm{B} \ket{s_\mathrm{B}}
 -|\la| \Tr \big| \Wsf^{\frac12}\Csf_\mathrm{B} \Fsf_\Brm \Csf_\mathrm{B}^\top\Wsf^{\frac12}  \big| \Big] .
\end{multline*}
\end{theorem}

\subsubsection{Maximum posterior variance bound for the qubit model}
To obtain the maximum posterior variance bound, we need to perform optimization over $\lambda$ of the bound given in Theorem \ref{thm:Clam}. This is done in Appendix \ref{sec:appC}. 
The possible optimal choices of $\lambda$ are $\la=0,\pm1$, or $\la_-$, which is defined below.  
Define the projector onto the orthogonal direction to $\ket{s_\mathrm{B}}$ by 
\begin{equation}
P_{\mathrm{B}}^\bot:=I-\frac{1}{|s_\mathrm{B}|^2} \ketbra{s_\mathrm{B}}{s_\mathrm{B}}, 
\end{equation}
for non-zero $\ket{s_\mathrm{B}}$. 
The ratio 
\begin{equation}
\Delta_\mathrm{B}(\Wsf):=
\begin{cases}
\infty& \mbox{for }\Tr \big| \Wsf^{\frac12}\Csf_\mathrm{B} \Fsf_\Brm \Csf_\mathrm{B}^\top\Wsf^{\frac12}  \big|=0\\[1.5ex]
\displaystyle\frac{\Tr[\Wsf^{\frac12}\Csf_\mathrm{B} P^\bot_\Brm \Csf_\mathrm{B}^\top\Wsf^{\frac12}]}{\Tr \big| \Wsf^{\frac12}\Csf_\mathrm{B} \Fsf_\Brm \Csf_\mathrm{B}^\top\Wsf^{\frac12}  \big| }& \mbox{otherwise}
\end{cases},
\end{equation}
plays a crucial role in the following discussion.

The maximum posterior variance bound is given as follows. 
The details of the classification are written in Appendix \ref{sec:appC}.
\begin{theorem} \label{thm:qubit_formula}
$\Ccal_\mathrm{B}^{(\la)}(\Wsf)$ is $\lambda$ independent if and only if $s_\mathrm{B}=0$. 
For non zero $s_\mathrm{B}\neq0$, we have three cases, 
\begin{equation*}
\Ccal_{\mathrm{B}}^{\max}(\Wsf)=\begin{cases}
\Ccal_\mathrm{B}^{(\la=0)}(\Wsf)& \mbox{ for } \Tr \big| \Wsf^{\frac12}\Csf_\mathrm{B} \Fsf_\Brm \Csf_\mathrm{B}^\top\Wsf^{\frac12}  \big|=0 \\[1ex]
\Ccal_\mathrm{B}^{(\la=\pm1)}(\Wsf)& \mbox{ for }\Delta_\mathrm{B}(\Wsf)\leq\frac12\left(1+\frac{1}{|s_\mathrm{B}|^2}\right) \\[1ex]
\Ccal_\mathrm{B}^{(\la=\la_-)}(\Wsf)& \mbox{ otherwise }
\end{cases},
\end{equation*}
where $\lambda_-$ is
\begin{equation}
\lambda_-:=\Delta_\mathrm{B}(\Wsf)-\sqrt{\Delta_\mathrm{B}(\Wsf)^2-\frac{1}{|s_\mathrm{B}|^2}}.
\end{equation}
\end{theorem}

\section{Examples}\label{sec:example}
In the example section, we analyze four models. 
In model A, three different optimal choices for the $\lambda$ are possible depending on the choice for the prior. 
In model B, the Bayesian SLD posterior variance bound is the best choice. In model C, on the other hand, the Bayesian RLD posterior variance bound is the optimal choice. The last example, model D, is a $\lambda$-independent one. 
To simplify the result, we set the weight matrix as $\Wsf=\Isf$ and denote the bound $\Ccal_\mathrm{B}^{(\la)}=\Ccal_\mathrm{B}^{(\la)}(I)$. The results are given in the main text and calculations will be given in Appendix \ref{sec:App_Exs}. 

\subsection*{Model A: The maximum posterior variance bound}
The scenario in which the maximum posterior variance bound exceeds both the Bayesian SLD and RLD posterior variance bounds provides a significant clue about the importance of this maximum posterior variance bound. This occurs in the standard two-parameter model where the linear Pauli model is characterized by a fixed parameter $\ep$ ($0 \leq \ep <  1$). 
(Negative $\ep$ case can be treated in the same manner.)
\begin{align*}
    \rho_\ta=\half (I+\ta^1 \sigma_1 +\ta^2 \sigma_2 +\ep \sigma_3),
\end{align*}
with the parameter space is $\Theta=\{ \ta \in \R^2 \,\big|\, |\ta|^2 \leq 1-\ep^2\}$. 
The prior is distributed independently as the cross item of \(\Csf_\pi\) is zero. If we assume the two parameters are independently and identically distributed (i.i.d.), the model is characterized by three numbers. These are the covariance matrix for the parameter and the expectation value
\begin{align*}
    \pi_1=\pi_2, ~ \Csf_\pi = \begin{pmatrix}
        v & 0 \\ 0 & v
    \end{pmatrix}, ~
    \ket{\mu} = \begin{pmatrix}
        \mu \\ \mu
    \end{pmatrix},
\end{align*}
and the parameter $\epsilon$. In the following, we analyze the maximum posterior variance bound as a function of $(v,\mu,\ep)$.
 
We take the calculation in Appendix~\ref{App:modelA}. The covariance matrix appears in both the numerator and denominator and cancels out. The result is shown as follows. 
\begin{equation*}
    \Ccal^{\max}_\mathrm{B} = 
    \begin{cases}
        \Ccal^{(\la=0)}_\mathrm{B} \quad &\mbox{for }\ep=0\\[1ex]
        \Ccal^{(\la=\pm 1)}_\mathrm{B}   \quad&\mbox{for }\ep>0 \mbox{ and } 2\mu^2 \leq \ep (1-\ep)\\[1ex]
        \Ccal^{(\la= \la_-)}_\mathrm{B}\quad   &\mbox{for }\ep>0 \mbox{ and } \ep(1-\ep) < 2\mu^2 < 1-\ep^2
    \end{cases},
\end{equation*}

and
\begin{align*}
    \la_- &= \frac{ \ep }{2 \mu^2 +\ep^2},\\
\Ccal_\mathrm{B}^{(\la=\la_-)} &= 2v- v^2 \left[3-\frac{2 \mu^2}{(2\mu^2 + \ep^2)(1- 2\mu^2 -\ep^2)}\right],\\
    \Ccal_\mathrm{B}^{(\la=1)} & = 2 v - 2 v^2 \frac{ 1-\ep  }{ 1-2\mu^2 -\ep^2 },\\
    \Ccal_\mathrm{B}^{(\la=0)} & =2v-2v^2\left[1+\frac{\mu^2}{1-2\mu^2-\ep^2} \right].
\end{align*}
Note that the optimal choice for $\lambda$ does not depend on the covariance $v$. 
This result allows for clear visualization of the maximum posterior variance bound in the $\mu,\ep$ plot. This is shown in Fig.~\ref{fig:max}. The red region is when the RLD posterior variance bound takes the lead, whereas the blue area is when $\la_-$ is the optimal solution which means the maximum posterior variance bound has a non-trivial $\lambda$ choice. Additionally, we provide the explanation of the boundaries, which are given by green lines. When $\ep=0$, the SLD posterior variance ($\la=0$) is always optimal. 
For $\mu=0$ with $\ep>0$, the condition $2\mu^2\le \ep(1-\ep)$ is automatically satisfied, hence the RLD posterior variance bound ($\la=1$) is optimal. 
The $\la$-independent case occurs only at $s_\mathrm{B}=0$, i.e., $(\mu,\ep)=(0,0)$.

\usepgfplotslibrary{fillbetween}
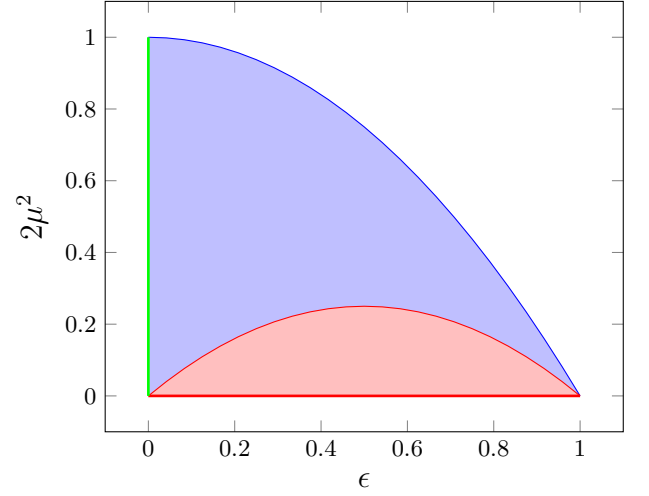
\begin{figure}[h]
\begin{center}
\begin{tikzpicture}[scale=1]
\begin{axis}[
    xlabel={{\large $\ep$}},
    ylabel={{\large $2\mu^2$}},
    title={}]
    \addplot [blue,name path=A,domain=0:1] {1-x*x};
 
    \addplot [red, name path=B,domain=0:1] {x*(1-x)};

    \addplot [red, line width = 1, name path=C,domain=0:1] {0};
    
    \addplot[green, line width = 1, samples=50, smooth,domain=0:1] coordinates {(0,0)(0,1)};

    \addplot [blue!50, fill opacity=0.5] 
        fill between[of=A and B];
    \addplot [red!50, fill opacity=0.5] 
        fill between[of=B and C];
\end{axis}
\end{tikzpicture}
\caption{Phase diagram for the optimal choice of maximum posterior variance bound. The blue region: $\la_-$ is optimal. The red region: RLD ($\la=1$) is optimal. The green lines: SLD ($\la=0$) is optimal.} \label{fig:max}
\end{center}
\end{figure}
We have a few remarks on the result based on the phase diagram in Fig.~\ref{fig:max}. 
\begin{itemize}
    \item For any fixed $\ep$, which means considering a line parallel to the $y$ axis, there exists a segment in the blue area at all $\ep$. This indicates the wide-ranging applicability of the maximum posterior variance bound.
    \item For any given $\mu$ with $2\mu^2\ge \frac{1}{4}$, whenever the point lies in the admissible region $2\mu^2<1-\ep^2$, we have $\la_-$ as the optimal choice. In other words, within the physical parameter region the maximum posterior variance bound dominates the SLD and RLD posterior variance bounds.
    \item On the other hand, when $2 \mu^2 < \frac{1}{4}$, the Bayesian RLD posterior variance bound is the preferred option.
    \item The Bayesian SLD posterior variance bound is not the optimal choice except for the special case when $\ep=0$.
    \item The expression of $\la_-$ is smooth at the boundary of the blue region. As $\ep \rightarrow 0$, $\la_- $ approaches 0, corresponding to the Bayesian SLD posterior variance bound. In addition, as $2 \mu^2 \rightarrow \ep(1-\ep)$, $\la_-$ approaches 1, representing the Bayesian RLD posterior variance bound. 
\end{itemize} 
We visualize the 3D envelope of the three candidates (SLD, RLD, and the $\lambda_-$ choice) in Fig.~\ref{fig:max3d}, and plot the corresponding gaps $C_{\mathrm{B}}^{(\lambda = \lambda_-)}-C_{\mathrm{B}}^{(\lambda=1)}$ and $C_{\mathrm{B}}^{(\lambda = \lambda_-)}-C_{\mathrm{B}}^{(\lambda=0)}$ in Figs.~\ref{fig:diffrld} and~\ref{fig:diffsld}, respectively.
\begin{figure}[!htb]
  \centering
  \includegraphics[width=\linewidth]{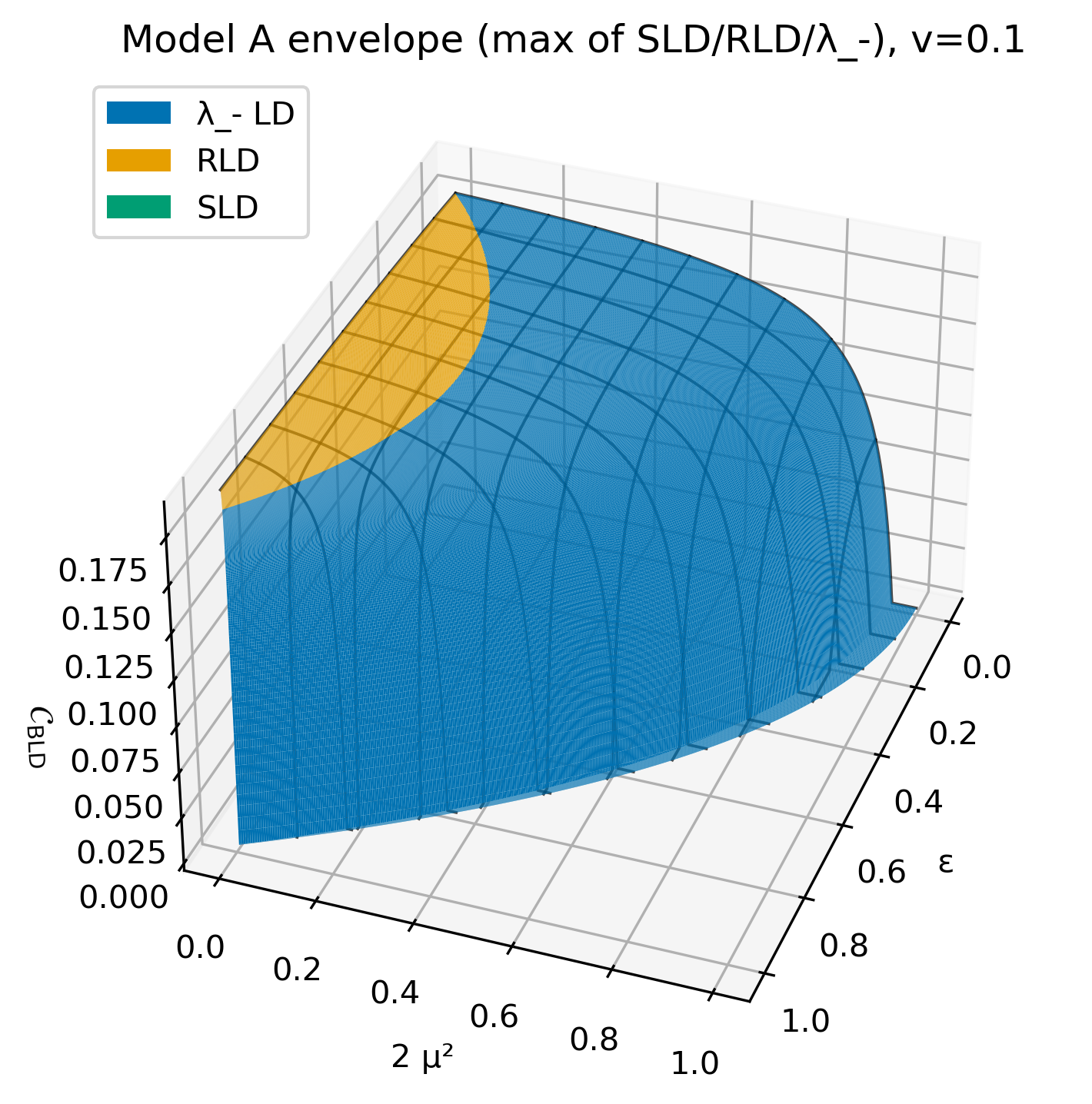}
   \caption{Model A: 3D envelope (pointwise maximum) of the three quantum posterior variance bounds.
The horizontal axes are $2\mu^2$ and $\epsilon$, and the vertical axis shows the envelope value
$C_{\mathrm{B}}^{\max}=\max\!\big\{C_{\mathrm{B}}^{(\lambda=0)},\,C_{\mathrm{B}}^{(\lambda=1)},\,C_{\mathrm{B}}^{(\lambda=\lambda_-)}\big\}$ (here $v=0.1$).
The surface is colored by the bound that attains the maximum at each $(2\mu^2,\epsilon)$:
blue for the $\lambda_-$ choice, orange for the RLD posterior variance bound ($\lambda=1$), and green for the SLD posterior variance bound ($\lambda=0$).}
\label{fig:max3d}
\end{figure}
\begin{figure}[!htb]
  \centering
  \includegraphics[width=\linewidth]{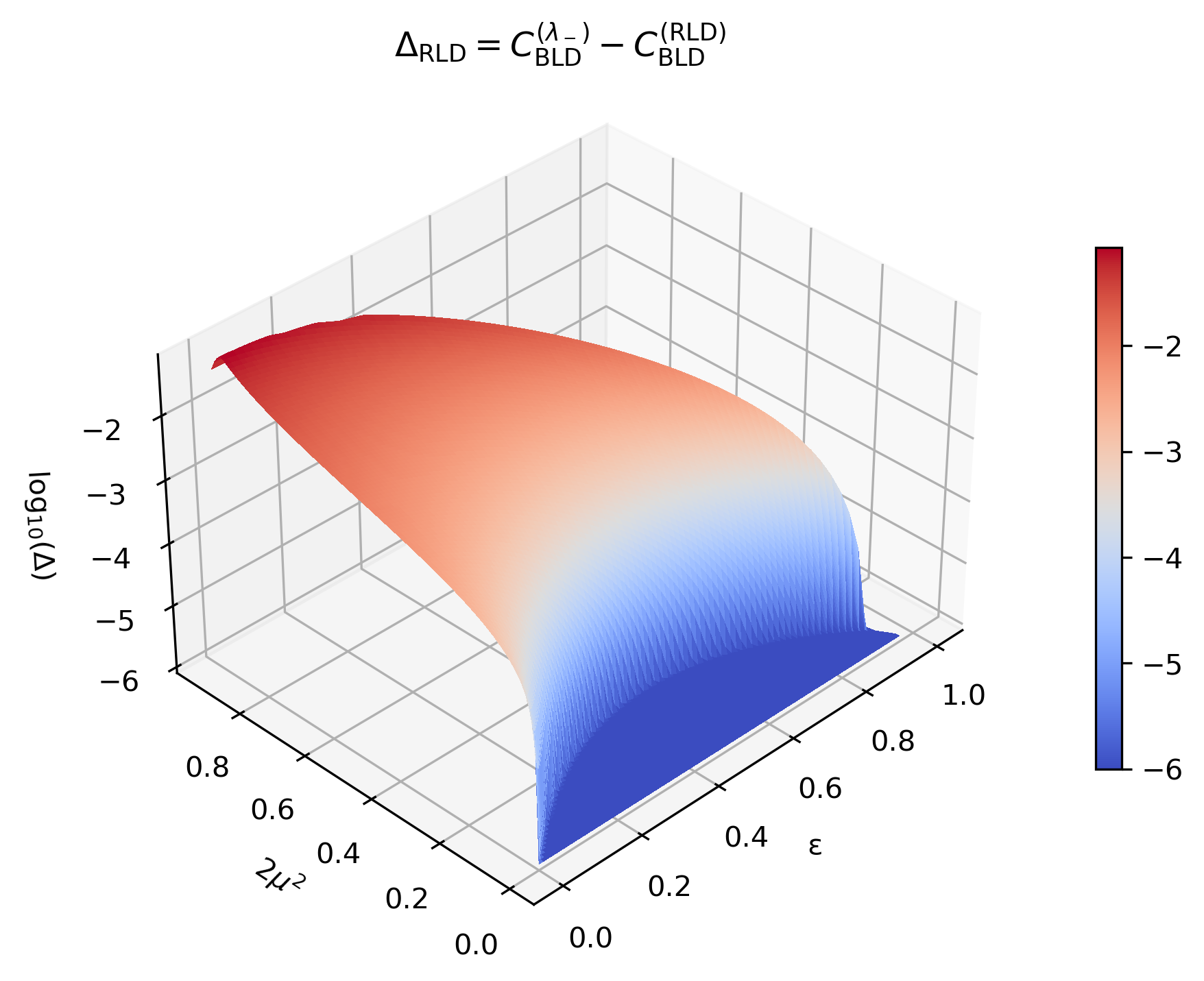}
   \caption{Model A: log-scale gap between the $\lambda_-$ and RLD posterior variance bounds.
The horizontal axes are $2\mu^2$ and $\epsilon$ (with $v=0.1$), and the plotted quantity is
$\log_{10}(\Delta_{\mathrm{RLD}})$, where
$\Delta_{\mathrm{RLD}} := C_{\mathrm{B}}^{(\lambda_-)} - C_{\mathrm{B}}^{(\mathrm{RLD})}$.
For visualization we clip $\Delta_{\mathrm{RLD}}$ from below at $10^{-6}$, i.e., we plot $\log_{10}(\max\{\Delta_{\mathrm{RLD}},10^{-6}\})$.
More negative values indicate a smaller improvement of the $\lambda_-$LD bound over the RLD posterior variance bound.}
  \label{fig:diffrld}
\end{figure}
\begin{figure}[!htb]
  \centering
  \includegraphics[width=\linewidth]{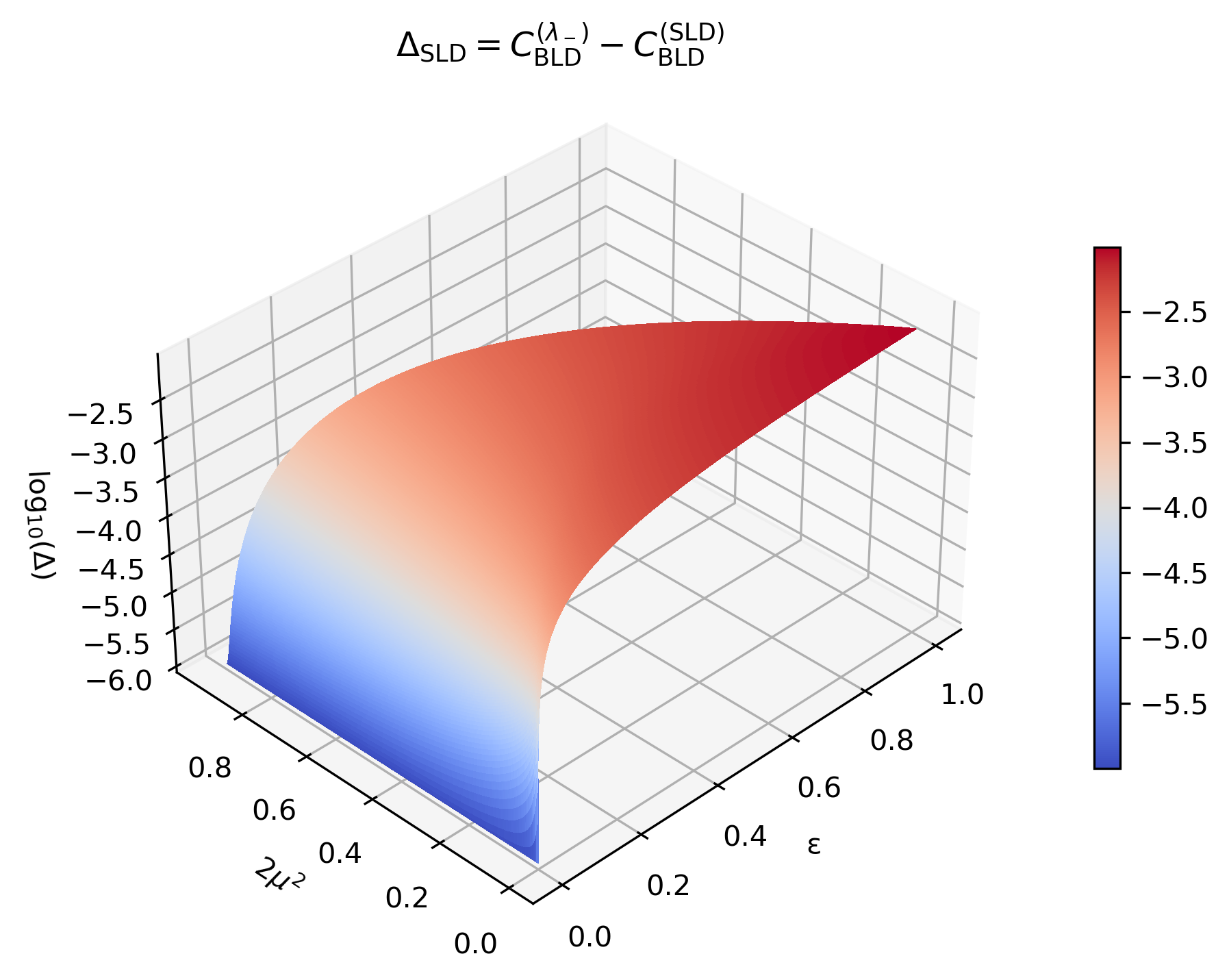}
   \caption{Model A: log-scale gap between the $\lambda_-$ and SLD posterior variance bounds.
The horizontal axes are $2\mu^2$ and $\epsilon$ (with $v=0.1$), and the plotted quantity is
$\log_{10}(\Delta_{\mathrm{SLD}})$, where
$\Delta_{\mathrm{SLD}} := C_{\mathrm{B}}^{(\lambda_-)} - C_{\mathrm{B}}^{(\mathrm{SLD})}$.
For visualization we clip $\Delta_{\mathrm{SLD}}$ from below at $10^{-6}$, i.e., we plot $\log_{10}(\max\{\Delta_{\mathrm{SLD}},10^{-6}\})$.
More negative values indicate a smaller improvement of the $\lambda_-$ posterior variance bound over the SLD posterior variance bound.}
  \label{fig:diffsld}
\end{figure}

\subsection*{Model B: The Bayesian SLD posterior variance bound}
In this part, we present a two-parameter qubit model in which the maximum posterior variance bound always coincides with the Bayesian SLD posterior variance bound. Consider the unitary model
\begin{align}
    \rho_\ta=\half\Bigl(I+\ta^1(\cos\ta^2\,\sigma_1+\sin\ta^2\,\sigma_2)\Bigr),
    \label{eq:modelB_state}
\end{align}
where $\ta^1\in(0,1)$ and $\ta^2\in[0,2\pi)$. The corresponding Bloch vector is
\begin{equation}
    s_\ta=\begin{pmatrix}\ta^1\cos\ta^2\\ \ta^1\sin\ta^2\\ 0\end{pmatrix},
\end{equation}
whose third component vanishes identically. 

We calculate the Bayesian $\lambda$-posterior variance bound in Appendix~\ref{App:modelB}. Here we only highlight the key structural reason why $\la=0$ is always optimal. 
Assume an independent prior
\begin{equation}
    \pi(\ta)=\pi_1(\ta^1)\pi_2(\ta^2).
\end{equation}
Let
\begin{equation}
    \ket{s_\mathrm{B}}:=\int_\Theta d\ta\,\pi(\ta)\,s_\ta
    =\begin{pmatrix}\mu_1 \csf\\ \mu_1 \ssf\\ 0\end{pmatrix},
    \qquad
    \mu^i:=\int_\Theta d\ta\,\pi(\ta)\,\ta^i,
\end{equation}
where we introduced the $\ta^2$-averages
\begin{align}
    \begin{pmatrix}\csf\\ \ssf\end{pmatrix}
    &:=\int_\Theta d\ta^2\,\pi_2(\ta^2)\begin{pmatrix}\cos\ta^2\\ \sin\ta^2\end{pmatrix}, \\
    \begin{pmatrix}\csf_2\\ \ssf_2\end{pmatrix}
    &:=\int_\Theta d\ta^2\,\pi_2(\ta^2)\ta^2\begin{pmatrix}\cos\ta^2\\ \sin\ta^2\end{pmatrix}.
\end{align}
Moreover, the matrix $\Csf_\mathrm{B}\in\mathbb{R}^{2\times 3}$ defined by
$\Csf_\mathrm{B}^{ij}:=\int_\Theta d\ta\,\pi(\ta)\,\ta^i s_\ta^j-\mu^i s_\mathrm{B}^j$
takes the explicit form
\begin{equation}
    \Csf_\mathrm{B}=
    \begin{pmatrix}
        \Csf_{\pi,11}\csf & \Csf_{\pi,11}\ssf & 0\\
        \mu_1(\csf_2-\mu_2\csf) & \mu_1(\ssf_2-\mu_2\ssf) & 0
    \end{pmatrix},
    \label{eq:modelB_CB}
\end{equation}
whose third column is identically zero. Since $s_\mathrm{B}^3=0$, the antisymmetric matrix $\Fsf_\mathrm{B}$ defined from $\ket{s_\mathrm{B}}$ has the structure
\begin{equation}
    \Fsf_\mathrm{B}=
    \begin{pmatrix}
        0&0&s_\mathrm{B}^2\\
        0&0&-s_\mathrm{B}^1\\
        -s_\mathrm{B}^2&s_\mathrm{B}^1&0
    \end{pmatrix}.
\end{equation}
As a consequence, we have
\begin{equation}
    \Csf_\mathrm{B}\Fsf_\mathrm{B}\Csf_\mathrm{B}^\top=\mathbf{0},
    \qquad
    \Tr\Bigl|\Wsf^{\frac12}\Csf_\mathrm{B}\Fsf_\mathrm{B}\Csf_\mathrm{B}^\top\Wsf^{\frac12}\Bigr|=0.
\end{equation}
Therefore, by the first case of Theorem~\ref{thm:qubit_formula}, the maximum posterior variance bound always equals the Bayesian SLD posterior variance bound:
\begin{align}
    \Ccal_\mathrm{B}^{\max}(\Wsf)=\Ccal_\mathrm{B}^{(\la=0)}(\Wsf).
\end{align}

Finally, using the explicit expression of $\Ksf_\Brm^{(\la=0)}$ given in Appendix~\ref{App:modelB}, the bound can be written in a compact closed form as follows. 
Introduce the $2\times 2$ matrix $\bar{\Csf}$ and the two-dimensional vector $\bar{s}$ by taking the first two components (since $s_\mathrm{B}^3=0$ in this model):
\begin{align}
    \bar{\Csf}&:=\!
    \begin{pmatrix}
        \Csf_{\pi,11}\csf&\Csf_{\pi,11}\ssf\\
        \mu_1(\csf_2-\mu_2\csf)&\mu_1(\ssf_2-\mu_2\ssf)
    \end{pmatrix},\
    \ket{\bar{s}} \! :=\!
    \begin{pmatrix}s_\mathrm{B}^1\\ s_\mathrm{B}^2\end{pmatrix}
    =\! \mu_1\begin{pmatrix}\csf\\ \ssf\end{pmatrix}.
\end{align}
Then
\begin{align}
    \Ccal_\mathrm{B}^{\max}(\Wsf)
    &=\Ccal_\mathrm{B}^{(\la=0)}(\Wsf)
    \\
    &=\Tr[ \Wsf \Csf_\pi]-\Tr[\bar{\Csf}^\top \Wsf \bar{\Csf}]
    -\frac{1}{1-\|\bar{s}\|^2}\,\bra{\bar{s}}\,\bar{\Csf}^\top \Wsf \bar{\Csf}\,\ket{\bar{s}}.
\end{align}

\subsection*{Model C: The Bayesian RLD posterior variance bound}
As another example, we consider a three-parameter qubit model in which the choice $\la=\pm1$ (the Bayesian RLD posterior variance bound) is always optimal. 
Let
\begin{align}
    \rho_\ta=\half\bigl(I+s_\ta\cdot\sigma\bigr),
    \qquad
    s_\ta=\begin{pmatrix}\ta^1\\ \ta^2\\ \ta^3\end{pmatrix},
    \qquad |\ta|<1 .
    \label{eq:modelC_state}
\end{align}
Assume an independent prior
\begin{equation}
    \pi(\ta)=\pi_1(\ta^1)\pi_2(\ta^2)\pi_3(\ta^3),
\end{equation}
and moreover $\pi_1=\pi_2=\pi_3$ (i.i.d.), so that the covariance matrix is isotropic:
\begin{equation}
    \Csf_\pi=v\,\Isf_3,\qquad v>0.
\end{equation}
Since $s_\ta=\ta$, the Bayesian mean Bloch vector is $\ket{s_\mathrm{B}}=\int_\Theta d\ta\,\pi(\ta)\, \ta$ and the matrix $\Csf_\mathrm{B}$ in Theorem~\ref{thm:qubit_formula} reduces to the covariance,
\begin{equation}
    \Csf_\mathrm{B}=\Csf_\pi=v\,\Isf_3.
\end{equation}

We take the isotropic weight $\Wsf=\Isf_3$. Then the quantity $\Delta_\mathrm{B}(\Wsf)$ defined in Theorem~\ref{thm:qubit_formula} can be evaluated explicitly (see Appendix~\ref{App:modelC} for details), yielding
\begin{equation}
    \Delta_\mathrm{B}=\frac{1}{\|s_\mathrm{B}\|}\qquad (\|s_\mathrm{B}\|>0).
\end{equation}
Therefore,
\[
\Delta_\mathrm{B}=\frac{1}{\|s_\mathrm{B}\|}\leq \frac12\left(1+\frac{1}{\|s_\mathrm{B}\|^2}\right)
\ \Longleftrightarrow\ 
\left(1-\frac{1}{\|s_\mathrm{B}\|}\right)^2\geq0,
\]
so the second condition in Theorem~\ref{thm:qubit_formula} is always satisfied whenever $\|s_\mathrm{B}\|>0$. 
Hence we are always in the second case of Theorem~\ref{thm:qubit_formula}, and the optimal choice is $\la=\pm1$ (both signs give the same value). 
Consequently, the maximum posterior variance bound coincides with the Bayesian RLD posterior variance bound:
\begin{align}
    \Ccal_\mathrm{B}^{\max}
    =\Ccal_\mathrm{B}^{(\la=\pm1)}
    =\Ccal_\mathrm{B}^{(\la=1)}
    = 3 v - \frac{3v^2 }{1-\|s_\mathrm{B}\|^2}  - 2 v^2 \frac{ 1 }{1 + \|s_\mathrm{B}\|} .
    \label{eq:modelC_Cmax}
\end{align}
For the special point $s_\mathrm{B}=0$, the bound is $\la$-independent by Theorem~\ref{thm:qubit_formula}, so the above conclusion remains valid.

Moreover, the corresponding matrix $\Ksf_\Brm^{(\la)}$ at the optimal choice $\la=1$ is given by
\begin{align}
    \Ksf_\Brm^{(\la=1)}
    = \ketbra{\mu}{\mu} +  \frac{v^2}{1-\|s_\mathrm{B}\|^2}\left(\Isf_3 - i \Fsf_\Brm \right),
    \label{eq:modelC_Klambda}
\end{align}
where $\Fsf_\Brm$ is the $3\times3$ antisymmetric matrix associated with $\ket{s_\mathrm{B}}$ as defined in \eqref{def:Fmatrix}.
(If desired, one may replace $\ketbra{\mu}{\mu}$ by $\ketbra{s_\mathrm{B}}{s_\mathrm{B}}$ since $s_\mathrm{B}=\mu$ in the present model.)

\subsection*{Model D: $\la$-independent case}
By Theorem~\ref{thm:qubit_formula}, $\Ccal_\mathrm{B}^{(\la)}(\Wsf)$ is $\la$-independent if and only if the Bayesian mean Bloch vector vanishes,
\begin{equation}
    \ket{s_\mathrm{B}}=\int_\Theta d\ta\,\pi(\ta)\,s_\ta=\mathbf{0},
\end{equation}
i.e., the averaged state $\bar{\rho}=\int_\Theta d\ta\,\pi(\ta)\rho_\ta$ is the maximally mixed state. 
Under this condition, the antisymmetric matrix $\Fsf_\mathrm{B}$ defined from $\ket{s_\mathrm{B}}$ is identically zero, $\Fsf_\mathrm{B}=\mathbf{0}$, and hence
\begin{equation}
    \Tr\Bigl|\Wsf^{\frac12}\Csf_\mathrm{B}\Fsf_\mathrm{B}\Csf_\mathrm{B}^\top\Wsf^{\frac12}\Bigr|=0.
\end{equation}
Therefore, we are in the first case of Theorem~\ref{thm:qubit_formula} and the Bayesian $\lambda$-posterior variance bound reduces to the $\la$-independent expression
\begin{align}
    \Ccal_\mathrm{B}^{(\la)}(\Wsf)
    =\Tr[\Wsf\,\Csf_\pi]-\Tr\!\big[ \Csf_\mathrm{B}^\top \Wsf\,\Csf_\mathrm{B}\big],
    \qquad (s_\mathrm{B}=0).
    \label{eq:modelD_lambda_indep}
\end{align}
In particular, any choice of $\la\in[-1,1]$ is equivalent (e.g., $\la=0$), and
\[
\Ccal_\mathrm{B}^{\max}(\Wsf)=\Ccal_\mathrm{B}^{(\la=0)}(\Wsf)=\Ccal_\mathrm{B}^{(\la)}(\Wsf).
\]

\section{Conclusion} \label{sec:conclusion}
We developed a metric-based approach to multiparameter quantum Bayesian estimation by introducing Bayesian monotone metrics, obtained by evaluating Petz monotone metrics on the prior-averaged state.
This construction transfers the metric viewpoint of quantum point estimation to the Bayesian setting while keeping the Bayes risk as the objective.

Within this framework, a choice of monotone metric (equivalently, an operator monotone function) specifies Bayesian analogues of quantum posterior mean operators and Fisher-information-type quantities: 
the quantum posterior-mean operators defined through a Bayesian operator equation, and the associated quantum Bayesian dual Fisher-information matrix \(\Ksf_{\mathrm B}\). 
These objects yield a transparent decomposition of the second-moment matrix into a metric-induced information contribution and admits a natural interpretation as a quantum posterior variance term. 
As a consequence, we obtained a systematic family of lower bounds, posterior variance bounds, on the Bayes risk.

We also established a universal ordering between two families of Bayesian lower bounds. 
For every monotone metric, the corresponding posterior variance bound dominates the associated quantum Bayesian Cram\'er--Rao bound, showing that the metric-induced posterior-variance formulation provides uniformly stronger performance limits in this Bayesian setting. 
Moreover, we resolved the main structural challenge of optimizing over the full (infinite-dimensional) family of operator monotone functions: the tightest bound can be obtained by restricting to a one-parameter \(\lambda\)-subfamily. 
This reduction turns the search for the best metric-induced bound into a tractable single-parameter optimization, leading to the maximum posterior variance bound as the best limit in the family.

Finally, through explicit qubit examples we demonstrated that the resulting bounds can be strictly tighter than the Bayesian SLD and RLD posterior variance bounds, with clear quantitative gaps. 
This highlights that importing monotone metrics into Bayesian quantum estimation is not merely a reformulation, but reveals genuinely stronger and systematically organized performance limits.

Several directions are open for future work.
A first question is to characterize attainability and tightness conditions for the maximum posterior variance bound, i.e., to identify when an explicit measurement and estimator can saturate (or asymptotically approach) the bound in multiparameter Bayesian settings.
Second, while the optimization over operator monotone functions reduces to a single real parameter \(\lambda\), it would be valuable to develop analytic criteria for the optimal \(\lambda\) and to obtain closed form solutions for broader classes of models beyond the examples studied here.
Third, the intrinsic metric-dependent term capturing multiparameter measurement incompatibility invites a deeper geometric interpretation, potentially linking Bayesian incompatibility to invariant quantities of the underlying geometry.
Finally, extending the framework to quantum processes and to noisy and multi-copy collective settings including adaptive and sequential Bayesian protocols may further broaden the operational scope of Bayesian monotone metrics and lead to new benchmarks in realistic finite-data regimes.

\section*{Acknowledgements}
This work was supported in part by ERATO ``Super Quantum Entanglement" (Grant No. JPMJER2402) from JST. KY is supported by JSPS KAKENHI Grant Numbers JP23H01090, JP22K03466. JS is supported by JSPS KAKENHI Grant Numbers JP21K11749, JP24K14816.

\bibliography{myref}

\appendix 

\section{Optimal estimator in classical setting}\label{sec:appBayes}
The optimal estimator is known to be given as follows. 
\begin{proposition}[Optimal Bayesian estimator]
For any prior and model, the optimal estimator and Bayes risk are given as
\begin{align*}
\arg\min_{\hattheta} \Rsf_\Brm[\hattheta|\Wsf]&=(\hattheta_{\Brm,i}) ,\\
\min_{\hattheta}\Rsf_\Brm[\hattheta|\Wsf]&=\Tr[\Wsf(\Msf-\Ksf_\Brm)]. 
\end{align*}
\end{proposition}
Remark: This result shows that the Bayesian MSE matrix $\Vsf_\Brm$ can be minimized as a matrix inequality. 
\begin{proof}
Define the expectation value of the parameter with respect to the joint distribution and posterior by 
\begin{align*}
d^i(x)&:=\int_\Theta d\ta\,\theta^iP(\ta,x),\\
\hattheta_{\Brm}^{i}(x)&:= \int_\Theta d\ta\,\theta^i\frac{P(\ta,x)}{P_X(x)}.
\end{align*} 
By definition, they are related by $\hattheta_{\Brm}^{i}P_X=d^i$. 
The Bayesian MSE matrix for any estimator $\hattheta$ is rewritten as
\begin{align}
  \Vsf_{\mathrm{B}}^{ij}[\hattheta]&=  \Ebb \left[ (\hattheta^i(X)-\theta^i) (\hattheta^j(X)-\theta^j) \right] \nonumber\\
  &= \Ebb [\hattheta^i(X)\hattheta^j(X)]-\Ebb[\ta^i \hattheta_j(X)]-\Ebb[\hattheta^i(X)\ta^j ]+\Msf^{ij} \nonumber\\
  &=\Ebb_X[\hattheta^i\hattheta^j]-\Ebb_X[\hattheta_{\Brm}^{i}\hattheta^j]-\Ebb_X[\hattheta^i \hattheta_{\Brm}^{j}]+\Msf^{ij}\nonumber\\
  &=\Ebb_X[(\hattheta^i-\hattheta_{\Brm}^{i})(\hattheta^j-\hattheta_{\Brm}^{j})]-\Ebb_X[\hattheta_{\Brm}^{i}\hattheta_{\Brm}^{j}]+\Msf^{ij}.\label{eq:appBayes1} 
\end{align}
To get the third line, we used 
\begin{align*}
\Ebb[\ta^i \hattheta^j(X)]&=\int_\Xcal dx\,d^i(x)\hattheta^j(x)\\
&=\int_\Xcal dx P_X(x) \frac{d^i(x)}{P_X(x)} \hattheta^j(x)\\
&=\int_\Xcal dx P_X(x) \hattheta_{\Brm}^{i} \hattheta^j(x)\\
&= \Ebb_X[\hattheta_{\Brm}^{i}(X)\hattheta^j(X)].
\end{align*}
Since the first term of Eq.~\eqref{eq:appBayes1} defines a positive semidefinite matrix, and it vanishes if and only if 
$\hattheta_i-\hattheta_{\Brm,i}$ for all $i$, we obtain
\begin{align*}
\Vsf_{\mathrm{B}}[\hattheta]\geq -\Ksf_\Brm+\Msf, 
\end{align*}
where $\Ksf_\Brm$ is defined by
\begin{align*}
\Ksf_{\Brm}^{ij}&= \Ebb_X[\hattheta_{\Brm}^{i}\hattheta_{\Brm}^{j}]\\
&=\int_\Xcal dx\, \frac{d_{\Brm}^{i}(x)d_{\Brm}^{j}(x)}{P_X(x)}. 
\end{align*}
\end{proof}

The van Trees lower bound is defined by the matrix $\Jsf_{\mathrm{B}}$. 
\begin{equation}
\Jsf_{\mathrm{B},ij}:= \Ebb\left[ \frac{\del\log P(\theta,X)}{\del \ta^i} \frac{\del\log P(\theta,X)}{\del \ta^j}\right]. 
\end{equation}
Note that this matrix is also expressed as 
\begin{equation}
\Jsf_{\mathrm{B},ij}=\Ebb_\pi\left[\frac{\del\log \pi(\ta)}{\del \ta^i}\frac{\del\log \pi(\ta)}{\del \ta^j}\right]
+\Ebb_\pi[\Jsf_{ij}(\ta)],
\end{equation}
where $\Jsf_{ij}(\ta)=\Ebb_\theta[\frac{\del\log p_\ta(x)}{\del \ta^i}\frac{\del\log p_\ta(x)}{\del \ta^j}]$ is the Fisher information about the model $p_\theta(x)$. 

Theorem \ref{thm:Cmain} states the matrix inequality: 
\begin{equation}
\Vsf_{\mathrm{B}}[\hattheta]\geq \Vmin=\Vsf_{\mathrm{B}}[\hattheta_\Brm]\geq (\Jsf_{\mathrm{B}})^{-1}. 
\end{equation}

To prove this result, we need a lemma. 
\begin{lemma}\label{lem:Bscore}
For a function $f$ of $(\theta, x)$ satisfying the weak boundary condition with respect to the prior, we have
\[
\Ebb\left[f(\theta,X) \frac{\del\log P(\theta,X)}{\del \ta^j}\right]=-\Ebb\left[\frac{\del f(\theta,X)}{\del \ta^j}\right]. 
\]
\end{lemma}
Remark: If $f$ is independent of the parameter $\ta$, it vanishes. 
If, on the other hand, $x$ independent, it is given by $-\Ebb_\pi[\frac{\del f(\theta)}{\del \ta^j}]$. 
\begin{proof}
The left hand side is calculated by using the integration by part as 
\begin{align*}
\mbox{(LHS)}&= \int_\Theta\int_\Xcal d\ta dx\, f(\theta,x) \frac{\del P(\theta,x)}{\del \ta^j}\\
&=\int_\Theta\int_\Xcal d\ta dx\, \frac{\del [f(\ta,x)P(\theta,x)]}{\del \ta^j}\\
&\quad- \int_\Theta\int_\Xcal d\ta dx\, P(\theta,x) \frac{\del f(\theta,x)}{\del \ta^j}\\
&=-\Ebb[\frac{\del f(\theta,X)}{\del \ta^j}].
\end{align*}
The last expression follows from that $f(\ta,x)P(\theta,X)$ along $\theta^j$ direction vanishes on $\del\Theta=$ the boundary of the parameter space $\Theta$. 
\end{proof}

\begin{proof}[Proof: Theorem \ref{thm:Cmain}]
Define an inner product for $F,G$ with respect to the joint distribution by
\begin{equation}
\inner{F}{G}_P:=\Ebb[F(\ta,X)G(\ta,x)]. 
\end{equation}
The optimal Bayesian MSE matrix is also expressed as 
\[
\Vsf_{\mathrm{B},ij}[\hattheta_\Brm]=\inner{\hattheta_{\Brm,i}-\theta^i}{\hattheta_{\Brm,j}-\theta^j}_P. 
\]
We introduce a $2n$-dimensional vector by
\begin{equation}
v(\ta,x):= \left[(\hattheta_{\Brm,i}(x)-\theta^i)_{i=1}^n,(\ell_{\Brm,j}(\ta,x))_{j=1}^n \right]^\top, 
\end{equation}
where $\ell_{\Brm,j}(\ta,x)=\frac{\del\log P(\theta,x)}{\del \ta^j}$. 
The expectation value of $v(\ta,X)v(\ta,X)^\top$ gives $2n\times2n$ positive matrix. 
We observe that its 11 block matrix is $\Vmin$, and the 22 block matrix is equal to $\Jsf_{\mathrm{B}}$. 
The off-diagonal block matrix is shown to be the identity matrix as follows.
\begin{align*}
\Ebb[(\hattheta_{\Brm,i}(X)&-\theta^i)\ell_{\Brm,j}(\ta,X)]\\
&=\Ebb[\hattheta_{\Brm,i}(X)\ell_{\Brm,j}(\ta,X)]-\Ebb[\theta^i \ell_{\Brm,j}(\ta,X)] \\
&=0-(-\frac{\del \ta^i}{\del\ta^j})\\
&=\delta_{ij},
\end{align*}
where we have used Lemma \ref{lem:Bscore}. 
Combining all gives
\[ 
\Ebb[v(\ta,X)v(\ta,X)^\top]=
\left[\begin{array}{cc} \Vmin& I \\ I & \Jsf_{\mathrm{B}}\end{array}\right]\geq 0. 
\]
By Schur's complement, we therefore obtain
\begin{equation}
\Vmin-I \Jsf_{\mathrm{B}}^{-1}I\geq0\ \iff\ \Vmin\geq \Jsf_{\mathrm{B}}^{-1}. 
\end{equation}
\end{proof}


\section{Bayesian Nagaoka--Hayashi bound} \label{App:BNH}

In this appendix, we recall the Bayesian Nagaoka--Hayashi (BNH) bound introduced in Ref.~\cite{suzuki2024bayesian}. 
For a POVM and estimator \(\hPi=(\Pi,\hattheta)\), define
\begin{align}
\mathbb{L}^{jk}[\hPi]
&:=\sum_{x\in\Xcal} \hat{\theta}^{j}(x)\Pi_x\hat{\theta}^{k}(x),
\quad j,k=1,\ldots,n,\\
X^j[\hPi]
&:=\sum_{x\in\Xcal}\hat{\theta}^{j}(x)\Pi_x,
\quad j=1,\ldots,n.
\end{align}
We regard \(\Lbb=[\mathbb{L}^{jk}[\hPi]]\) as an operator-valued matrix on \(\Cbb^n\otimes\Hil\), and
\[
X[\hPi]=[X^1[\hPi],X^2[\hPi],\ldots,X^n[\hPi]]^{\top_1}
\]
as an operator-valued vector, where \({\top_1}\) denotes the transpose over the parameter space. 
The model and prior distribution define
\begin{align}
    \rho_\Brm&:=\int_\Theta d\ta\, \pi(\ta) \rho_\ta, \\
    D_{\Brm}&:=[D^j_{\Brm}],\qquad
    D^j_{\Brm}:=\int_\Theta d\ta\, \pi(\ta) \ta^j \rho_\ta,\\
    \Msf&:=[\Msf^{jk}],\qquad
    \Msf^{jk}:=\int_\Theta d\ta\, \pi(\ta)\ta^j\ta^k .
\end{align}

The first ingredient is an operator-valued representation of the Bayes risk.

\begin{lemma}[Suzuki \cite{suzuki2024bayesian}, Lemma 2]\label{lem:Briskrep}
For any POVM and estimator \(\hPi=(\Pi,\hattheta)\), the Bayes risk is expressed as
\begin{align}
 \Rsf_\mathrm{B}[\hPi|\Wsf]
 &= \Ttr{\left[(\Wsf \ot \rho_\Brm) \Lbb\right]}
  -\Ttr{\left[(\Wsf \ot I) D_\Brm  X^{\top_1}\right]}
  \nonumber\\
 &\quad
  -\Ttr {\left[X  D_\Brm^{\top_1}(\Wsf \ot I)\right] }
  + \Tr[\Wsf\Msf] .
\end{align}
Here \(\Ttr[\cdot]\) denotes the trace over both the Hilbert space and the parameter space.
\end{lemma}

The second ingredient is a matrix inequality for \(\Lbb\) and \(X\), which follows from a variant of Holevo's lemma \cite{holevobook}.

\begin{lemma}[\cite{hayashi99}]
For all POVMs and estimators, the corresponding operator-valued matrix \(\Lbb\) and vector \(X\) satisfy
\begin{align}
    \Lbb[\hPi] \geq X[\hPi] X[\hPi]^{\top_1}.
\end{align}
\end{lemma}

Combining these two lemmas, one obtains the following Bayesian Nagaoka--Hayashi bound.

\begin{theorem}[Bayesian Nagaoka--Hayashi bound, Suzuki \cite{suzuki2024bayesian}, Theorem 1]\label{thm:BNHbound}
For any POVM and estimator \(\hPi=(\Pi,\hattheta)\), the Bayes risk satisfies
\begin{align}
    \Rsf_\mathrm{B}[\hPi|\Wsf]\geq \cNH(\Wsf),
\end{align}
where
\begin{align}
  \cNH(\Wsf)
  :=
  \min_{\Lbb,X}
  \Big\{
  &\Ttr{\left[(\Wsf \ot \rho_\Brm) \Lbb\right]}
  -\Ttr{\left[(\Wsf \ot I) D_\Brm  X^{\top_1}\right]}
  \nonumber\\
  &-\Ttr{\left[X  D_\Brm^{\top_1}(\Wsf \ot I)\right]}
  \Big\}
  + \Tr[\Wsf\Msf].
\end{align}
The optimization is subject to the constraints that
\(\Lbb^{jk}=\Lbb^{kj}\), each \(\Lbb^{jk}\) is Hermitian, each \(X^j\) is Hermitian, and
\begin{align}
    \Lbb\geq X X^{\top_1}.
\end{align}
\end{theorem}

Indeed, the above optimization is obtained by replacing the original optimization over measurements and estimators \((\Pi,\hattheta)\) with an optimization over the larger feasible set of operator-valued variables \((\Lbb,X)\). 
This relaxation yields a lower bound on the optimal Bayes risk:
\begin{align}
    \Rsf^{\mathrm{opt}}_\Brm(\Wsf)\geq \cNH(\Wsf).
\end{align}
Moreover, as shown in Ref.~\cite{suzuki2024bayesian}, this optimization can be formulated as a semidefinite program.

\section{Monotone metric on matrix space}\label{sec:AppOpMono}
In this appendix, we summarize known results on monotone metric. 
The material in this appendix is due to Ref.~\cite{petz1996monotone}. 

\subsection{Monotone metric}
Let $M_d=\Cbb^{d\times d}$ denote the set of all $d\times d$ complex matrices, and define $M_d^{++}$ $(M_d^{+})$ be the set of 
all positive definite (positive semidefinite) matrices. 
The set of all Hermitian matrices is denoted by $M_{d}^{H}$. 
The set of full-rank quantum states is defined by
\[
S_d:=\{\rho\in M_d^{++}\,|\,\tr[\rho]=1 \}.
\]

Given $A,B\in M_d^{++}$, and $f:(0,\infty)\to(0,\infty)$ define a super-operator $\Jbb^f_{A,B}$ acting on $X\in M_d$ by 
\[
\Jbb^f_{A,B} (X)= f(\Delta_{A,B})X B, 
\]
with $\Delta_{A,B}(X)=A XB^{-1}$. 
It is positive operator with respect to the Hilbert-Schmidt inner product, and hence is invertible. 
The inverse super-operator is denoted by $(\Jbb_{A,B}^f)^{-1}$. 
When $A=B$, we simply denote it as $\Jbb^f_{A}=\Jbb^f_{A,A}$. 

A monotone metric on $S_d$ is defined by three conditions. 
Among them, the second condition states that norm should not be increased under any CP-TP map. 
\begin{definition}[Monotone metric] \label{def:monotone_metric}
An inner product $K_\rho$ ($\rho\in M_d$) on $M_d\times M_d$ is said monotone metric, if the following conditions are satisfied. 
\begin{itemize}
\item[(i)] The map $\rho\mapsto K_\rho(X,X)$ is continuous for any $d\in \Nbb$ and $X\in M_d$. 
\item[(ii)] For any CP-TP map $\Ecal$ from $M_d$ to $M_{d'}$ such that $\Ecal(I_d)>0$, then 
the following inequality holds 
\[
K_{\Ecal(\rho)} (\Ecal(X),\Ecal(X))\leq K_\rho(X,X) \mbox{ for }\rho\in S_d,\,X\in M_d.
\]
\item[(iii)] If $[\rho,X]=0$, $K_\rho(X,X)=\tr[\rho^{-1} X^\dagger X]$. 
\end{itemize}
\end{definition}
Remark: In order to recover the results in classical statistics, we demand the condition $f(1)=1$ which follows from condition (iii). 

The next theorem due to Petz characterizes all possible monotone metrics. 
\begin{theorem}[Petz]\label{thm:petz_thm}
An inner product $K$ is monotone metric if and only if 
there exists a (unique) operator monotone function $f:(0,\infty)\to(0,\infty)$ such that 
\[
K_\rho(X,Y)=\inner{X}{(\Jbb^f_\rho)^{-1}Y}_{\mathrm{HS}}.
\]
for $\rho\in M_d^{++}$ and $X,Y\in M_d$. 
\end{theorem}

Since operator monotone functions play an important role in matrix analysis, there exist many equivalent characterization. 
Among them the next theorem is important in this paper. 
\begin{theorem}[Hiai-Petz]\label{thm:HP_thm}
Given $f:(0,\infty)\to(0,\infty)$, the following conditions are equivalent. 
\begin{itemize}
\item[(i)] $f$ is operator monotone function. 
\item[(ii)] For all $X\in M_d$, $(A,B)\mapsto \inner{X}{\Jbb^f_{A,B}X}_{\mathrm{HS}}$ is jointly concave on $M_d^{++}\times M_d^{++}$. 
\item[(iii)] For all $X\in M_d$, $(A,B)\mapsto \inner{X}{(\Jbb^f_{A,B})^{-1}X}_{\mathrm{HS}}$ is jointly convex on $M_d^{++}\times M_d^{++}$.   
\end{itemize}
\end{theorem}

An important consequence of the Hiai-Petz theorem is that monotone metric satisfies concavity and convexity. 
This is done by setting $\Jbb^f_{\rho,\rho}=\Jbb^f_{\rho}$ in Theorem \ref{thm:HP_thm}.
\begin{corollary}\label{cor:conv-concav}
For any monotone metric, the following properties hold. 
\begin{itemize}
\item[(i)] $\inner{X}{\Jbb^f_{\rho}X}_{\mathrm{HS}}$ is concave with respect to $\rho\in S_d$ and convex with respect to $X\in M_d$. 
\item[(ii)] $\inner{X}{(\Jbb^f_{\rho})^{-1}X}_{\mathrm{HS}}$ is jointly convex with respect to $\rho\in S_d$ and $X\in M_d$. 
\end{itemize}
\end{corollary}

\section{Concavity of the $f$-quadratic form}\label{app:concavity}
In this appendix we justify the following fact used in the proof of Theorem~\ref{thm:ineq2}:
for any fixed operator $X$, the map
\[
\rho \ \longmapsto\ \Phi_f(\rho;X):=\tr\!\left[X^\dagger \mathbb J_\rho^f(X)\right]
\]
is concave, where $\mathbb J_\rho^f$ is the Kubo--Ando mean of the left and right multiplication super-operators
associated with the operator monotone function $f$ with $f(1)=1$.

\subsection{Preliminaries}
For $\rho>0$, define the left/right multiplication super-operators
\[
L_\rho(Y):=\rho Y,\qquad R_\rho(Y):=Y\rho,
\]
and the modular operator $\Delta_\rho:=L_\rho R_\rho^{-1}$.
Given an operator monotone function $f:(0,\infty)\to\mathbb{R}$ with $f(1)=1$,
define
\begin{equation}\label{eq:Jf_def}
\mathbb J_\rho^f \ :=\ R_\rho^{1/2}\, f(\Delta_\rho)\, R_\rho^{1/2}.
\end{equation}
Equivalently, for positive definite operators
\(A>0\) and \(B>0\), using the \emph{mean transformation} (the operator perspective)
\begin{equation}\label{eq:Pf_def}
P_f(A,B)\ :=\ B^{1/2}\, f\!\left(B^{-1/2}AB^{-1/2}\right)\, B^{1/2},
\end{equation}
one has the identity
\begin{equation}\label{eq:Jf_as_Pf}
\mathbb J_\rho^f \;=\; P_f(L_\rho,R_\rho).
\end{equation}

We will use two standard ingredients:
(i) operator monotone $f$ on $(0,\infty)$ is operator concave;
(ii) Jensen's operator inequality for operator concave functions.

\subsection{Joint concavity of the mean transformation}
\begin{lemma}[Jensen inequality]\label{lem:Jensen}
Let $f$ be operator concave on $(0,\infty)$.
If $\{K_\alpha\}$ satisfies $\sum_\alpha K_\alpha^\dagger K_\alpha = I$ and each $X_\alpha>0$,
then
\[
f\!\left(\sum_\alpha K_\alpha^\dagger X_\alpha K_\alpha\right)
\ \geq\ \sum_\alpha K_\alpha^\dagger f(X_\alpha) K_\alpha.
\]
\end{lemma}

\begin{lemma}[Joint concavity of $P_f$]\label{lem:Pf_joint_concave}
Let $f$ be operator concave on $(0,\infty)$.
Then $P_f(A,B)$ defined in \eqref{eq:Pf_def} is jointly concave on pairs $(A,B)$ with $B>0$:
for any $t\in[0,1]$ and $A_1,A_2 > 0$, $B_1,B_2>0$,
\[
P_f(A_t, B_t) ~\geq~ t P_f(A_1,B_1) + (1-t)P_f(A_2,B_2),
\]
\end{lemma}
where $A_t:=tA_1+(1-t)A_2$ and $B_t:=tB_1+(1-t)B_2$ (so $B_t>0$).

\begin{proof}
Let
\begin{align*}
X_i:=B_i^{-1/2}A_iB_i^{-1/2}\quad(i=1,2),
\\
K_1:=\sqrt{t}\,B_1^{1/2}B_t^{-1/2},\ \ 
K_2:=\sqrt{1-t}\,B_2^{1/2}B_t^{-1/2}.  
\end{align*}
Then
\[
K_1^\dagger K_1 + K_2^\dagger K_2
= B_t^{-1/2}(tB_1+(1-t)B_2)B_t^{-1/2} = I,
\]
and a direct computation shows
\begin{align*}
K_1^\dagger X_1 K_1 + K_2^\dagger X_2 K_2
&= B_t^{-1/2}(tA_1+(1-t)A_2)B_t^{-1/2}\\
&= B_t^{-1/2}A_tB_t^{-1/2}.  
\end{align*}
Applying Lemma~\ref{lem:Jensen} gives
\[
f(B_t^{-1/2}A_tB_t^{-1/2})
\geq
K_1^\dagger f(X_1)K_1 + K_2^\dagger f(X_2)K_2.
\]
Multiplying $B_t^{1/2}$ from left and right yields
\begin{align*}
&B_t^{1/2} f(B_t^{-1/2}A_tB_t^{-1/2}) B_t^{1/2}
\\  \geq 
& t\,B_1^{1/2} f(X_1) B_1^{1/2} + (1-t)\,B_2^{1/2} f(X_2) B_2^{1/2},  
\end{align*}
which is exactly the desired joint concavity of $P_f$.
\end{proof}

\subsection{Concavity of $\rho \mapsto \Phi_f(\rho;X)$}
\begin{proposition}[Concavity of the $f$-quadratic form]\label{prop:Phi_concave}
Let $f$ be operator monotone on $(0,\infty)$ with $f(1)=1$.
Fix an operator $X$. Then the map
\[
\rho \ \longmapsto\ \Phi_f(\rho;X)=\tr\!\left[X^\dagger \mathbb J_\rho^f(X)\right]
\]
is concave on the set of density operators $\rho>0$.
\end{proposition}

\begin{proof}
Since $f$ is operator monotone on $(0,\infty)$, it is operator concave.
By Lemma~\ref{lem:Pf_joint_concave}, $P_f(A,B)$ is jointly concave in $(A,B)$.

Now note that $L_\rho$ and $R_\rho$ depend \emph{linearly} on $\rho$, i.e.,
for $\rho_t:=t\rho_1+(1-t)\rho_2$,
\[
L_{\rho_t}=tL_{\rho_1}+(1-t)L_{\rho_2},\qquad
R_{\rho_t}=tR_{\rho_1}+(1-t)R_{\rho_2}.
\]
Therefore, using \eqref{eq:Jf_as_Pf} and the joint concavity of $P_f$,
\begin{align*}
\mathbb J_{\rho_t}^f
= P_f(L_{\rho_t},R_{\rho_t})
&\geq tP_f(L_{\rho_1},R_{\rho_1}) + (1-t)P_f(L_{\rho_2},R_{\rho_2})
\\
&= t\mathbb J_{\rho_1}^f + (1-t)\mathbb J_{\rho_2}^f.  
\end{align*}
Finally, applying the superoperator inequality
to the vector $X$ in the Hilbert--Schmidt space,
we obtain
\begin{align*}
\Phi_f(\rho_t;X)
&=
\tr\!\left[
X^\dagger
\mathbb J_{\rho_t}^f(X)
\right]
\\
&\ge
\tr\!\left[
X^\dagger
\bigl(
t\mathbb J_{\rho_1}^f
+
(1-t)\mathbb J_{\rho_2}^f
\bigr)(X)
\right]
\\
&=
t\,\Phi_f(\rho_1;X)
+
(1-t)\,\Phi_f(\rho_2;X),
\end{align*}
which proves concavity. 
\end{proof}
This result also could be noticed from Theorem~\ref{thm:petz_thm}, Theorem~\ref{thm:HP_thm} and Corollary~\ref{cor:conv-concav}.

\section{Proofs}\label{sec:App_proof}

\subsection{Proof for Corollary \ref{cor:fBLD}} \label{sec:App_proof1}
Here is the well-known lemma (see, for example, Lemma 6.6.1 in Ref.~\cite{holevobook}) to prove the corollary. 
Given positive semidefinite matrices $\Zsf,\Wsf\in M_d^{+}$ such that $\Wsf$ is real symmetric, consider the following optimization:
\[
F(\Zsf|\Wsf)=\min_{\Vsf}\{\Tr[\Wsf\Vsf]\,|\,\Vsf\mbox{: real symmetric},\Vsf\geq \Zsf\}. 
\]
This has an analytical form as follows. 
\begin{lemma} \label{lem:equiv_holevo}
For any $\Zsf,\Wsf\geq0$, we have 
\[
F(\Zsf|\Wsf)=\Tr[\Wsf\Zsf]+\Tr|\Wsf^{1/2}\Imag \Zsf\Wsf^{1/2}|, 
\]
and the optimizer $\Vsf_*$ is 
\[
\Vsf_*=\Wsf^{1/2}\Real \Zsf\Wsf^{1/2}+\Wsf^{-1/2}|\Wsf^{1/2}\Imag Z\Wsf^{1/2}|\Wsf^{-1/2},
\]
where $|X|:=\sqrt{X^\dagger X}$ is the absolute operator. 
\end{lemma}
Using this lemma, we prove the corollary as follows. 
(Here the constraint r.s. stands for real symmetric.)
\begin{align*}
\Rsf[\hPi|\Wsf]&=\Tr[\Wsf \Vsf[\hPi]]\\
&\geq \min_{\Vsf:\mbox{r.s.}}\{\Tr[\Wsf \Vsf]| \Vsf\geq\Msf-\Ksf^f_\Brm\}\\
&=F(\Msf-\Ksf^f_\Brm|\Wsf)\\
&=\Tr[\Wsf(\Msf-\Ksf^f_\Brm)]+\Tr|\Imag (\Msf-\Ksf^f_\Brm)|\\
&=\Tr\left[\Wsf \Msf\right]-\Tr\left[\Wsf \Real\Ksf^f_\Brm\right]+\Tr\left|\Wsf^{\frac12} \Imag\Ksf^f_\Brm\Wsf^{\frac12} \right| . 
\end{align*}
In the last line, $\Imag (\Msf-\Ksf^f_\Brm)=-\Imag \Ksf^f_\Brm$ is used. $\square$

\subsection{Proof for Eq.~\eqref{eq:cqBayes}} \label{sec:App_proof_cq}
To prove Eq.~\eqref{eq:cqBayes}, we note that the monotonicity of the Bayesian metric \eqref{eq:monotone}. 
Let $\Hcal_\Xcal$ be the Hilbert space whose dimension $\dim \Hcal_\Xcal =|\Xcal|=:d_\Xcal$ is same as the measurement outcome set $\Xcal$. 
As a special case, we choose the following CP-TP map from $S_d$ to $S_{d_\Xcal}$. 
\begin{equation}
\Ecal:\ \rho\mapsto \ket{x}e_x(\rho)\bra{x} \mbox{ with } e_x(\rho)= \tr[\rho \Pi_x]. 
\end{equation}
Its domain can be extended to all \(d \times d\) complex matrices $M_d$. 
Note any two matrices $A,B$ will commute after applying the map as $[\Ecal(A), \Ecal(B)]=0$, diagonal matrices commute with each other. 
By taking arbitrary complex vector $c=(c_i)\in\Cbb^n$, monotonicity Eq.~\eqref{eq:monotone} for any $X\in M_{d}$ yields
\begin{align*}
c^\dagger \Ksf^f_{\Brm} c &= \sum_{i,j}c_i^* K^f_{\rhoB}(D^{i}_\Brm,D^{j}_{\Brm}) c_j\\
&=K^f_{\rhoB}(D_c,D_c) \quad \mbox{ where }D_c:=\sum_i c_iD^{i}_{\Brm} \\
&\geq K^f_{\Ecal(\rhoB)}(\Ecal(D_c),\Ecal(D_c))\quad[\mbox{monotonicity}]\\
&= \tr\left[\Ecal(\rhoB)^{-1} \Ecal(D_c)^\dagger \Ecal(D_c) \right]\quad[\mbox{(iii) of Def.~\ref{def:monotone_metric}}]\\
&=\sum_{x\in\Xcal} \frac{e_x(D_c)^* e_x(D_c)}{e_x(\rhoB)}\\
&=\sum_{i,j}c_i^* \sum_{x\in\Xcal} \frac{e_x(D^{i}_\Brm) e_x(D^{j}_\Brm)}{e_x(\rhoB)} c_j\\
&=c^\dagger \Ksf^{f}_{\Brm}[\Pi] c. \hspace{3cm} \square
\end{align*}

\subsection{Proof for Theorem \ref{thm:ineq1}} \label{sec:App_proof2}
The original proof can be found in Ref.~\cite{yamagata2021maximum} and we give a short version of it here. 
It is known that $f:(0,\infty)\to(0,\infty)$ is operator monotone, if and only if $f$ is operator concave. 
By the correspondence to the classical Fisher information matrix (condition (iii) of Definition \ref{def:monotone_metric}) when $f(1)=1$ needs to be satisfied. Since $f$ is differentiable, we can always find $\la_c\in[-1,1]$ such that 
\[
f(t)\leq f_{\la_c}(t)\mbox{ with }\la_c=2 f'(1)-1. 
\]
for all $t\in(0,\infty)$. 
This implies for any $X\in M_d$ and $\rho\in M_d^{++}$, we have
\[
\inner{X}{(\Jbb_\rho^{f})^{-1}X}_{\mathrm{HS}}\geq\inner{X}{(\Jbb_\rho^{f_c})^{-1}X}_{\mathrm{HS}}. 
\]
This is equivalent to the statement: 
for any operator monotone function $f$ there always exists $\lambda$ such that 
\[
K_\rho^f(X,X)\geq K_\rho^{(\la)}(X,X). 
\]

\subsection{Proof for Theorem \ref{thm:ineq2}}

\subsubsection{Eq.~\eqref{eq:step1}} \label{sec:App_step1}
We prove
\[
\Msf-\Ksf^f_\Brm\geq 
\left[\inner{E^{f,i}_{\Brm}-\theta^i I}{E^{f,j}_{\Brm}-\theta^j I}^{f}_{\pi(\ta),\rho_\ta} \right].
\]
The right hand side consists of four terms. We evaluate one by one. 
Let $c\in\Cbb^d$ be an arbitrary vector, and define
\begin{align*}
E^{f}_c&:=\sum_i c_iE^{f,i}_{\Brm},\\
\theta_c&:=\sum_i c_i\theta^{i}.
\end{align*} 
\begin{align} \label{eq:step1-1}
\sum_{i,j}c_i^*\inner{E^{f,i}_{\Brm}}{E^{f,j}_{\Brm}}^{f}_{\pi(\ta),\rho_\ta}c_j 
&=\int_\Theta d\ta \pi(\ta) \inner{E^{f}_{c}}{E^{f}_{c}}^{f}_{\rho_\ta} \nonumber\\
&\leq c^\dagger \Ksf^f_\Brm c.\ [\mbox{\, Equation~(\ref{eq:concavity_core})}]
\end{align}
\begin{align}\label{eq:step1-2}
\sum_{i,j}c_i^*\inner{E^{f,i}_{\Brm}}{\ta^{j}I}^{f}_{\pi(\ta),\rho_\ta}c_j 
&=\inner{E^{f}_{c}}{\ta_cI}^{f}_{\pi(\ta),\rho_\ta}\nonumber\\
&=\Ebb_\pi\left[ \inner{E^{f}_{c}}{\Jbb^f_{\rho_\ta}\ta_{c}}_{\mathrm{HS}}\right]\nonumber\\
&=\Ebb_\pi\left[ \inner{E^{f}_{c}}{\rho_\ta\ta_{c}}_{\mathrm{HS}}\right]\nonumber\\
&= \inner{E^{f}_{c}}{\Ebb_\pi[\rho_\ta\ta_{c}]}_{\mathrm{HS}}\nonumber\\
&=\inner{E^{f}_{c}}{D_{c}}_{\mathrm{HS}}\nonumber\\
&=\sum_{ij} c_i^* \inner{E^{f,i}_{\Brm}}{D^j_{\Brm}}_{\mathrm{HS}}c_j\nonumber\\
&=\sum_{ij} c_i^* \Ksf^{f,ij}_{\Brm}c_j
\end{align}
Similarly, we have
\begin{equation} \label{eq:step1-3}
\sum_{i,j}c_i^*\inner{\ta^{i}I}{E^{f,j}_{\Brm}}^{f}_{\pi(\ta),\rho_\ta}c_j 
=\sum_{ij} c_i^* \Ksf^{f,ij}_{\Brm}c_j.
\end{equation}
Lastly, 
\begin{align} \label{eq:step1-4}
\sum_{i,j}c_i^*\inner{\ta^{i}I}{\ta^{i}I}^{f}_{\pi(\ta),\rho_\ta}c_j 
&=\Ebb_\pi\left[ \inner{\ta_{c}}{\rho_\ta\ta_{c}}_{\mathrm{HS}}\right]\nonumber\\
&=\sum_{ij} c_i^* \Ebb_\pi[\ta^i\ta^j]c_j\nonumber\\
&=c^\dagger \Msf c.
\end{align}
Combining Eqs.~(\ref{eq:step1-1},\ref{eq:step1-2},\ref{eq:step1-3},\ref{eq:step1-4}) gives 
\begin{align*}
\sum_{i,j}c_i^*\inner{E^{f,i}_{\Brm}-\theta^i I&}{E^{f,j}_{\Brm}-\theta^j I}^{f}_{\pi(\ta),\rho_\ta}c_j \\
&=\inner{E^{f}_{c}-\ta_c I}{E^{f}_{c}-\ta_c I}^{f}_{\pi(\ta),\rho_\ta}\\
&\leq c^\dagger \Ksf^f_\Brm c-2 c^\dagger \Ksf^f_\Brm c+c^\dagger \Msf_ c\\
&=c^\dagger (\Msf-\Ksf^f_\Brm) c
\end{align*}

\subsubsection{Lemma} \label{sec:App_step2}
\begin{lemma} \label{lem:QBscore}
For any matrix $X_\ta$ which depends on $\theta$ satisfying the weak boundary condition with respect to the prior, we have
\[
\left\langle X_\ta\,,\,L^{f}_{\ta,j}+\frac{\del\log \pi(\ta)}{\del \ta^j} I\right\rangle^{f}_{\pi(\ta),\rho_\ta}=-\Ebb_\pi\left[\tr\left[\rho_\ta\frac{\del X_\theta}{\del \ta^j}\right]\right]. 
\]
\end{lemma}
\begin{proof}\
Let $X_\ta$ be any matrix which may depend smoothly on the parameter $\ta\in\Theta$. 
By the assumption of the weak boundary condition, it satisfies 
\[
\pi(\ta) \tr[\rho_\ta X_\ta]\to 0 \mbox{ at }\theta\in\del\Theta, 
\]
for any direction $\ta^i$ ($i=1,2,\ldots,n$). 
The left hand side of the lemma is written by using the definition of $ L^{f}_{\ta,j}$ as follows.
\begin{align}
(\mathrm{LHS})&=\Ebb_\pi[\inner{X_\ta}{L^{f}_{\ta,j}}^f_{\rho_\ta}]+\Ebb_\pi[\del_j\log\pi(\ta)\inner{X_\ta}{I}^f_{\rho_ta}]\nonumber\\
&=\Ebb_\pi[\inner{X_\ta}{\del_j\rho_\ta}_{\mathrm{HS}}]+\Ebb_\pi[\del_j\log\pi(\ta)\tr[X_\ta\rho_\ta]]\nonumber\\
&=\Ebb_\pi[\tr[X_\ta\del_j\rho_\ta]]+\Ebb_\pi[\del_j\log\pi(\ta)\tr[X_\ta\rho_\ta]]. \label{eq:lem_QBscore}
\end{align}
Thus two terms are independent of the choice $f$. 
The second term is evaluated using integration by part and the assumption of the weak boundary condition as
\begin{align*}
\Ebb_\pi[\del_j\log\pi(\ta)\tr[X_\ta\rho_\ta]]
&=\int_\Theta d\ta \pi(\ta) \del_j\log\pi(\ta)\tr[X_\ta\rho_\ta]\\
&=\int_\Theta d\ta  \del_j\pi(\ta)\tr[X_\ta\rho_\ta]\\
&=-\int_\Theta d\ta  \pi(\ta)\del_j\tr[X_\ta\rho_\ta]\\
&=- \Ebb_\pi[\del_j\tr[X_\ta\rho_\ta]]. 
\end{align*}
Combining the last expression with the first term of Eq.~\eqref{eq:lem_QBscore}, we get 
\begin{align*}
(\mathrm{LHS})&=\Ebb_\pi[\tr[X_\ta\del_j\rho_\ta]]- \Ebb_\pi[\del_j\tr[X_\ta\rho_\ta]]\\
&=- \Ebb_\pi[\tr[X_\ta\del_j\rho_\ta]]. 
\end{align*}
This proves the lemma. 
\end{proof}

\subsubsection{Eq.~\eqref{eq:step3}} \label{sec:App_step3}
Consider the two sets of operators. 
\begin{align}
\mbox{set 1:}&\quad\left\{E^{f,i}_{\Brm}-\theta^i I \right\}_{i=1}^n,\\ 
\mbox{set 2:}&\quad\left\{ L^{f}_{\ta,j}+\frac{\del\log \pi(\ta)}{\del \ta^j} I\right\}_{j=1}^n.  
\end{align}
We define $2n\times 2n$ matrix as in the proof of Theorem \ref{thm:Cmain}. 
The 11 block matrix is given by
\begin{align*}
\left\langle E^{f,i}_{\Brm}-\theta^i I\,,\,E^{f,j}_{\Brm}-\theta^j I \right\rangle^{f}_{\pi(\ta),\rho_\ta}
&=\Vsf^f_{\Brm,ij}. 
\end{align*}
The 12 block and 21 block matrices are the identity matrix due to Lemma \ref{lem:QBscore}. 
\begin{align*}
&\left\langle E^{f,i}_{\Brm}-\theta^i I\,,\, L^{f}_{\ta,j}+\frac{\del\log \pi(\ta)}{\del \ta^j} I\right\rangle^{f}_{\pi(\ta),\rho_\ta}\\
&\qquad=-\Ebb_\pi\left[\tr[\rho_\ta \del_j(E^{f,i}_{\Brm}-\theta^i I) ]\right]\\
&\qquad=\Ebb_\pi[\tr[\rho_\ta]\delta^i_{j}]\\
&\qquad=\delta^i_{j} 
\end{align*}
The 22 block matrix is identical to the Bayesian Fisher information matrix as 
\begin{align*}
&\left\langle L^{f}_{\ta,i}+\frac{\del\log \pi(\ta)}{\del \ta^j} I\,,\,L^{f}_{\ta,j}+\frac{\del\log \pi(\ta)}{\del \ta^j} I\right\rangle^{f}_{\pi(\ta),\rho_\ta}\\
&=\Ebb_\pi[\inner{L^{f}_{\ta,i}}{L^{f}_{\ta,j}}^f_{\rho_\ta}]+\Ebb[\del_i \log\pi(\ta)\del_j \log\pi(\ta)]\inner{I}{I}^f_{\rho_\ta}\\
&=\Jsf_{\Brm,ij}^f,
\end{align*}
where we used $\inner{I}{L^{f}_{\ta,j}}^f_{\rho_\ta}=\tr[\del_j\rho_\ta]=0$ and $\inner{I}{I}^f_{\rho_\ta}=1$. 

Therefore, the matrix under consideration is 
\[
\left[\begin{array}{cc} \Vsf^{f}_\Brm & \Isf \\ \Isf & \Jsf^{f}_{\mathrm{B}}\end{array}\right]. 
\]
By definition, this is positive semi-definite, and hence Schur's complement gives 
\[
\Vsf^f_\Brm-\Isf (\Jsf_{\mathrm{B}}^f)^{-1}\Isf \geq0\ \iff\ \Vsf^f_\Brm\geq (\Jsf_{\mathrm{B}}^f)^{-1}. 
\]
\subsection{Proof for $\cNH(\Wsf)\geq \Ccal^{\max}_\mathrm{B}(\Wsf)$} \label{sec:App_BNH-BLD} 
This inequality is proven by showing 
\begin{equation} \label{eq:NHtoBLD}
\cNH(\Wsf)\geq\clamBLD(\Wsf),
\end{equation}
for any $\la\in[-1,1]$. Taking the maximum over $\lambda$, we get the desired statement. 

To prove inequality \eqref{eq:NHtoBLD}, we need to introduce another lower bound, the Bayesian Holevo-type bound \cite{suzuki2024bayesian}.
\begin{definition}[Bayesian Holevo-type bound] 
\begin{multline} \label{def:BHbound}
\cH(\Wsf):=\min_{\Vsf,X=(X^1,\ldots,X^n)}\{\Tr[\Wsf \Vsf]\,|\,\Vsf:\mathrm{r.s.},X^i:\mathrm{Hermitian},\\
\Vsf\geq \Zsf[X]-\Hsf[X]-\Hsf[X]^\top+\Msf \},
\end{multline}
where $\Zsf^{ij}[X]:=\tr[\rhoB X^iX^j]$ and $\Hsf^{ij}[X]:=\tr[D^i_\Brm X^j]$. 
\end{definition}
Using Lemma \ref{lem:equiv_holevo}, it has another form:
\begin{multline*}
\cH(\Wsf)=\min_{X=(X^1,\ldots,X^n)}\Tr[\Wsf (\Zsf[X]-\Hsf[X]-\Hsf[X]^\top)]\\+\Tr|\Wsf^{1/2}\Imag \Zsf[X]\Wsf^{1/2}|+\Tr[\Wsf \Msf],
\end{multline*}
where $X^i$ are Hermitian. In the following discussion, we use the two expressions concurrently. 
In contrast to the BNH bound, the Bayesian Holevo-type bound is expressed only in terms of optimization of $X=(X^1,\ldots, X^n)$. 
The inequality \eqref{eq:NHtoBLD} is proven by the following two inequalities. 
\begin{proposition}
For any weight matrix and any $\la\in[-1,1]$, three bounds $\cNH,\cH,\clamBLD$ satisfy
\begin{align}
\cNH(\Wsf)&\geq\cH(\Wsf), \label{eq:NHtoH}\\
\cH(\Wsf)&\geq\clamBLD(\Wsf) \label{eq:HtoBLD}. 
\end{align}
\end{proposition}
\begin{proof}
The first inequality is proven in Ref.~\cite{suzuki2024bayesian}, and we only give a sketch of it. 
Consider a variable $\Lbb$ in the feasible set of the BNH bound. 
The partial trace with respect to the quantum system yields a $n\times n$ real positive semi-definite matrix. 
This is because $\Lbb^{ij}=\Lbb^{ji}:$ Hermitian implies $\Vsf:=\tr_{\Hcal}[\Lbb]$ is real symmetric. 
The other constraint $\Lbb\geq XX^{\top_1}\geq0$ shows $\Vsf$ is positive semi-definite. 
Therefore, all $\Lbb$ in the feasible set of the BNH bound are in the feasible set of optimization \eqref{def:BHbound}. 
Hence, $\cH(\Wsf)$ is always smaller or equal to $\cNH(\Wsf)$. 

Next, let us prove the second inequality \eqref{eq:HtoBLD}. 
We do this by two steps as follows. 
Define the one-parameter family of the Bayesian Holevo-type bound $\cH^{(\la)}(\Wsf)$ ($\la\in[-1,1]$) by
\begin{multline*}
\cH^{(\la)}(\Wsf)=\min_{X=(X^1,\ldots,X^n)}\Tr[\Wsf (\Zsf^{(\la)}[X]-\Hsf[X]-\Hsf[X]^\top)]\\+\Tr|\Wsf^{1/2}\Imag \Zsf^{(\la)}[X]\Wsf^{1/2}|+\Tr[\Wsf \Msf],
\end{multline*}
where $\Zsf^{(\la),ij}[X]:=\inner{X^i}{X^j}^{(\la)}_{\rhoB}$ and $X^i$ are Hermitian. 
Note the choice $\la=\pm1$ reduces to the Bayesian Holevo-type bound. 
Then, we will show below that  
\begin{align*}
\cH^{(\la=\pm1)}(\Wsf)&\geq\cH^{(\la)}(\Wsf),\\
\cH^{(\la)}(\Wsf)&\geq \clamBLD(\Wsf),
\end{align*}
for any $\la\in[-1,1]$. 
The first inequality is straightforward, since the imaginary part of $\Zsf^{(\la)}$ is proportional to $\lambda$ and $|\la|\leq1$. 
That is $\Tr|\Wsf^{1/2}\Imag \Zsf^{(\la)}[X]\Wsf^{1/2}|=|\la| \Tr|\Wsf^{1/2}\Imag \Zsf^{(\la=1)}[X]\Wsf^{1/2}|$. 
Thus, we only need to show the second inequality. 
The essential point here is to rewrite the matrix $\Hsf$ in terms of the quantum posterior-mean operator as
\[
\Hsf^{ij}[X]=\tr[D^i_\Brm X^j]=\inner{E^{(\la),i}_\Brm}{X^j}^{(\la)}_{\rhoB}. 
\] 
This allows us to complete the square for the matrix inequality of Eq.~\eqref{def:BHbound} as 
\begin{align*}
&\Zsf^{(\la),ij}[X]-\Hsf^{ij}[X]-\Hsf^{ji}[X]+\Msf^{ij}\\
&=\inner{X^i}{X^j}^{(\la)}_{\rhoB}-\inner{E^{(\la),i}_\Brm}{X^j}^{(\la)}_{\rhoB}-\inner{X^i}{E^{(\la),j}_\Brm}^{(\la)}_{\rhoB}\\
&=\inner{X^i-E^{(\la),i}_\Brm}{X^j-E^{(\la),j}_\Brm}^{(\la)}_{\rhoB}-\inner{E^{f,i}_{\Brm}}{E^{f,j}_{\Brm}}^f_{\rho_\Brm}+\Msf^{ij}. \\
&=\inner{X^i-E^{(\la),i}_\Brm}{X^j-E^{(\la),j}_\Brm}^{(\la)}_{\rhoB}+\Msf^{ij}-\Ksf^{(\la),ij}_\Brm. 
\end{align*}
In general, $E^{(\la),j}_\Brm$ are not Hermitian and there does not exist optimal $X$ to make the first matrix zero. 
Nevertheless, we can extend the feasible set of the bound $\cH^{(\la)}(\Wsf)$ to any matrix and this gives a lower bound. 
\begin{align*}
&\cH^{(\la)}(\Wsf)
\geq \min_{X=(X^1,\ldots,X^n)}\Tr[\Wsf (\Zsf^{(\la)}[X]-\Hsf[X]-\Hsf[X]^\top)]\\
&\hspace{2cm}+\Tr|\Wsf^{1/2}\Imag \Zsf^{(\la)}[X]\Wsf^{1/2}|\\
&=\min_{\Vsf,X=(X^1,\ldots,X^n)} \{\Tr[\Wsf \Vsf]\,|\, \Vsf:\mathrm{r.s.},\Vsf \geq \Zsf^{(\la)}[X]\\
&\hspace{4.5cm}-\Hsf[X]-\Hsf[X]^\top+\Msf\}\\
&=\min_{\Vsf:\mathrm{r.s.}}\{\Tr[\Wsf \Vsf]\,|\, \Vsf \geq \Msf-\Ksf^{(\la)}_\Brm\}\\
&=\clamBLD(\Wsf). 
\end{align*}
\end{proof}

\section{$\lambda$-posterior variance bound in General qubit form}\label{sec:appC}
In this appendix, we give the calculation of classification for $\lambda$-posterior variance bound optimal for the general qubit model.  
Note that the expression $\Ksf^{(\la)}$, Eq.~\eqref{eq:qubitK}, is directly solvable, since it is a $2\times2$ matrix space. Theorem \ref{thm:Clam} then follows immediately. 
Let us analyze the quantum posterior variance bound $\Ccal_\mathrm{B}^{(\la)}(\Wsf)=:f(\la)$ as a function of $\lambda$. 
Due to symmetry, we restrict $\la\geq0$ case only. 

First of all, $\Ccal_\mathrm{B}^{(\la)}(\Wsf)$ is $\lambda$ independent if and only if $\Ebb_\pi[\theta^i]=0$ for all $i$, which is equivalent to the condition $s_\Brm=0$. 
This is shown as follows. 
$\lambda$ independence is possible if and only if the matrix $\Ksf^{(\la)}$ does not depend on $\lambda$. 
From the expression of one line above of Eq.~\eqref{eq:qubitK}, we have the condition,
\begin{align}
& \frac{\partial}{\partial \lambda} \Csf_\mathrm{B}\left[ I - \ketbra{s_\mathrm{B}}{s_\mathrm{B}} +i \la \Fsf_\mathrm{B} \right]^\inv \Csf_\mathrm{B}^\top=0\\
\Leftrightarrow & 
\frac{\partial}{\partial \lambda} \left( I - \ketbra{s_\mathrm{B}}{s_\mathrm{B}} +i \la \Fsf_\mathrm{B} \right)=0. 
\end{align}
This is easily solved as $s_\mathrm{B}=0$. This proves the first statement of Theorem \ref{thm:qubit_formula}. 

Next, we shall consider the case $s_\mathrm{B}\neq0$. \\
\textit{SLD-type optimal}: By taking the derivative with respect to $\lambda \geq 0$, we have 
\begin{equation*}
\frac{\del f(\la)}{\del \la}=\frac{|s_\Brm|^2}{(1-\la^2|s_\Brm|^2)^2}\left[a_F \la^2-2 \Tr{[P_\mathrm{B}^\bot\Csf_\mathrm{B}^\top\Wsf \Csf_\mathrm{B}]}\la+\frac{a_F}{|s_\Brm|^2} \right], 
\end{equation*}
where $a_F:=\Tr \big| \Wsf^{\frac12}\Csf_\mathrm{B} \Fsf_\Brm \Csf_\mathrm{B}^\top\Wsf^{\frac12}  \big| $. 
We observe that if $a_F$ is zero, $f(\la)$ is monotonically decreasing and hence the maximum occurs at $\la_*=0$. 
This is the first case of Theorem \ref{thm:qubit_formula}. 

To proceed, we consider the case $a_F\neq0$. 
When $a_F\neq0$, the behavior of the function $F(\la)$ is determined by 
the quadratic function $g(\la):=\la^2 -2 \Delta_\Brm (\Wsf)\la+\frac{1}{|s_\Brm|^2}$. 
The roots of $g(\la)=0$ are shown to be real and positive. 
\[
\la_\pm:=\Delta_\Brm (\Wsf)\pm\sqrt{\Delta_\Brm (\Wsf)^2-\frac{1}{|s_\Brm|^2}}. 
\]
It is sufficient to show $\Delta_\Brm (\Wsf)-\frac{1}{|s_\Brm|}\geq0$ (see the last part of this Appendix.).\\ 
\textit{RLD-type optimal}: The choice $\la=1$ is optimal if $f(\la)$ is monotonically increasing. 
This is possible if and only if $\la_-\geq1$, which is equivalent to 
\begin{equation*}
\la_- \geq 1 \iff \Delta_\Brm (\Wsf)\leq\frac12\left(1+\frac{1}{|s_\mathrm{B}|^2}\right). 
\end{equation*}

Finally the generic case of the optimal choice occurs when $\la_-< 1$. 
To show this, we note that the function $f(\la)$ is increasing from $\la\in[0,\la_-]$ and then decreasing up to $\la=1$. 
This is because $\la_+\geq \Delta_\Brm (\Wsf)\geq \frac{1}{|s_\Brm|}>1$. 
We then conclude that the maximum occurs at $\la=\la_-$. 

\textit{Proof for} $\Delta_\Brm (\Wsf)-\frac{1}{|s_\Brm|}\geq0$: \\
To prove this inequality, we consider an Hermitian matrix: 
\[
\Zsf:=I-\ket{\bar{s}_\Brm}\bra{\bar{s}_\Brm}+i \Fsf_\Brm/|s_\Brm|=P^\bot_\Brm+i\Fsf_\Brm/|s_\Brm|, 
\]
where $\bar{s}_\Brm={s}_\Brm/|s_\Brm|$ is the normalized vector. 
Note that the matrix $\Zsf$ is positive semidefinite, since its eigenvalues are 
all non-negative; $0,0,2$ (rank 1). 
Then, we have the relation $\Tr[\Real \Zsf]\geq \Tr|\Imag \Zsf|$. 
Similarly, we extend it to 
\[
\Zsf(\Wsf):=  \Wsf^{\frac12}\Csf_\mathrm{B} \Zsf \Csf_\mathrm{B}^\top\Wsf^{\frac12}, 
\]
and since \(\Zsf (\Wsf) \geq 0 \), we obtain the inequality $\Tr[\Real \Zsf(\Wsf)]\geq \Tr|\Imag \Zsf(\Wsf)|$. 
The rest is to rewrite it as follows. 
\begin{align*}
&\Tr[\Real \Zsf(\Wsf)]\geq \Tr|\Imag \Zsf(\Wsf)| \\
&\Leftrightarrow \Tr[\Wsf^{\frac12}\Csf_\mathrm{B} P^\bot_\Brm \Csf_\mathrm{B}^\top\Wsf^{\frac12}]\geq 
\frac{1}{|s_\Brm|}\Tr \big| \Wsf^{\frac12}\Csf_\mathrm{B} \Fsf_\Brm \Csf_\mathrm{B}^\top\Wsf^{\frac12}  \big|\\
&\Leftrightarrow\Delta_\Brm (\Wsf)=\frac{\Tr[\Wsf^{\frac12}\Csf_\mathrm{B} P^\bot_\Brm \Csf_\mathrm{B}^\top\Wsf^{\frac12}]}{\Tr \big| \Wsf^{\frac12}\Csf_\mathrm{B} \Fsf_\Brm \Csf_\mathrm{B}^\top\Wsf^{\frac12}  \big| }\geq \frac{1}{|s_\Brm|}.
\end{align*}

\section{Calculations for examples}\label{sec:App_Exs}
\subsection{Model A}\label{App:modelA}

We consider the two-parameter qubit model
\[
  \rho_{\ta}=\frac12\Bigl(I+\ta^1\sigma_1+\ta^2\sigma_2+\ep\sigma_3\Bigr),
  \qquad
  \Theta=\{\ta\in\mathbb{R}^2\mid |\ta|^2\le 1-\ep^2\}.
\]
The prior is i.i.d. with mean and covariance
\[
  \ket{\mu}=\begin{pmatrix}\mu\\ \mu\end{pmatrix},
  \qquad
  \Csf_\pi=\begin{pmatrix}v&0\\0&v\end{pmatrix},
  \qquad (v>0).
\]

\paragraph{Bloch vectors and Bayesian averages.}
The Bloch vector is
\[
  s_{\ta}=\begin{pmatrix}\ta^1\\ \ta^2\\ \ep\end{pmatrix}.
\]
Hence the Bayesian mean Bloch vector is
\[
  \ket{s_\mathrm{B}}
  :=\int_\Theta d\ta\,\pi(\ta)\, s_{\ta}
  =\begin{pmatrix}\mu\\ \mu\\ \ep\end{pmatrix},
  \qquad
  \|s_\mathrm{B}\|^2=2\mu^2+\ep^2.
\]
Throughout this appendix we assume the physical condition $\|s_\mathrm{B}\|^2<1$.

\paragraph{The matrix $\Csf_\mathrm{B}$.}
By definition,
\[
  \Csf_\mathrm{B}^{ij}
  :=\int_\Theta d\ta\,\pi(\ta)\,\ta^i s_{\ta}^j-\mu^i s_\mathrm{B}^j,
  \qquad (i=1,2;\ j=1,2,3),
\]
and using $\mathbb{E}[\ta^i]=\mu$ and $\mathrm{Cov}(\ta^i,\ta^k)=v\,\delta_{ik}$, we obtain
\[
  \Csf_\mathrm{B}
  =\begin{pmatrix}
    v&0&0\\
    0&v&0
  \end{pmatrix}.
\]

\paragraph{The matrix $\Fsf_\mathrm{B}$.}
With $\ket{s_\mathrm{B}}=(s_\mathrm{B}^1,s_\mathrm{B}^2,s_\mathrm{B}^3)^\top=(\mu,\mu,\ep)^\top$,
\[
  \Fsf_\mathrm{B}
  =\begin{pmatrix}
    0 & -s_{\mathrm{B}}^{3} & s_{\mathrm{B}}^{2} \\
    s_{\mathrm{B}}^{3} & 0 & -s_{\mathrm{B}}^{1} \\
    -s_{\mathrm{B}}^{2} & s_{\mathrm{B}}^{1} & 0 
  \end{pmatrix}
  =
  \begin{pmatrix}
    0 & -\ep & \mu\\
    \ep & 0 & -\mu\\
    -\mu & \mu & 0
  \end{pmatrix}.
\]

\paragraph{The key antisymmetric term in Theorem~\ref{thm:qubit_formula}.}
For $\Wsf=I_2$,
\[
  \Csf_\mathrm{B}\Fsf_\mathrm{B}\Csf_\mathrm{B}^\top
  =
  v^2\begin{pmatrix}0&-\ep\\ \ep&0\end{pmatrix},
  \qquad
  \Tr\Bigl|\Csf_\mathrm{B}\Fsf_\mathrm{B}\Csf_\mathrm{B}^\top\Bigr|
  =2v^2|\ep|.
\]
Therefore the first case of Theorem~\ref{thm:qubit_formula} occurs iff $\ep=0$ (for $v>0$).

\paragraph{The quantity $\Delta_\mathrm{B}(\Wsf)$ and the phase boundary.}
Let
\[
  P_\perp:=I_3-\frac{\ket{s_\mathrm{B}}\bra{s_\mathrm{B}}}{\|s_\mathrm{B}\|^2}.
\]
For $\Wsf=I_2$ we have
\[
  \Csf_\mathrm{B}^\top\Wsf\,\Csf_\mathrm{B}
  =\Csf_\mathrm{B}^\top\Csf_\mathrm{B}
  =\mathrm{diag}(v^2,v^2,0),
\]
and
\[
  \Tr\!\left[P_\perp\,\Csf_\mathrm{B}^\top\Csf_\mathrm{B}\right]
  =2v^2\,\frac{\mu^2+\ep^2}{2\mu^2+\ep^2}.
\]
Combining with $\Tr|\Csf_\mathrm{B}\Fsf_\mathrm{B}\Csf_\mathrm{B}^\top|=2v^2|\ep|$, for $\ep>0$ we obtain the closed form
\[
  \Delta_\mathrm{B}(I_2)
  =\frac{\Tr\!\left[P_\perp\,\Csf_\mathrm{B}^\top\Csf_\mathrm{B}\right]}
         {\Tr\bigl|\Csf_\mathrm{B}\Fsf_\mathrm{B}\Csf_\mathrm{B}^\top\bigr|}
  =\frac{\mu^2+\ep^2}{\ep(2\mu^2+\ep^2)}.
\]
Moreover, the condition in Theorem~\ref{thm:qubit_formula}
\[
  \Delta_\mathrm{B}(I_2)\le \frac12\Bigl(1+\frac{1}{\|s_\mathrm{B}\|^2}\Bigr)
\]
is equivalent (for $\ep>0$) to the simple inequality
\[
  2\mu^2 \le \ep(1-\ep).
\]
This is the red/blue phase boundary used in Fig.~\ref{fig:max}.

\paragraph{The optimal interior value $\la_-$.}
In the remaining region $\ep(1-\ep)<2\mu^2<1-\ep^2$ (with $\ep>0$), Theorem~\ref{thm:qubit_formula} yields
\[
  \la_-=\Delta_\mathrm{B}(I_2)-\sqrt{\Delta_\mathrm{B}(I_2)^2-\frac{1}{\|s_\mathrm{B}\|^2}}
  =\frac{\ep}{2\mu^2+\ep^2}.
\]

\paragraph{Explicit $K_\mathrm{B}^{(\la)}$ for Model A.}
Restricting the qubit $\la$-family to the two parameters $(\ta^1,\ta^2)$ and evaluating at the Bayesian mean state
$\bar\rho=\frac12(I+s_\mathrm{B}\cdot\sigma)$, we obtain
\[
  K_\mathrm{B}^{(\la)}
  =
  \begin{pmatrix}
    1+\dfrac{\mu^2}{1-\|s_\mathrm{B}\|^2} &
    \dfrac{\mu^2-i\la\,\ep}{1-\|s_\mathrm{B}\|^2} \\
    \dfrac{\mu^2+i\la\,\ep}{1-\|s_\mathrm{B}\|^2} &
    1+\dfrac{\mu^2}{1-\|s_\mathrm{B}\|^2}
  \end{pmatrix},
  \qquad
  \|s_\mathrm{B}\|^2=2\mu^2+\ep^2.
\]
(Equivalently, $\Real K_\mathrm{B}^{(\la)}$ depends only on $\mu$, and the $\la$-dependence appears only through
$\Imag K_\mathrm{B}^{(\la)}\propto \la\,\ep\begin{pmatrix}0&-1\\1&0\end{pmatrix}$.)
Substituting these ingredients into the general expression of $\Ccal_\mathrm{B}^{(\la)}(\Wsf)$ gives the closed forms
reported in Model~A in the main text.

\subsection{Model B}\label{App:modelB}

We consider the two-parameter unitary model
\[
  \rho_\ta=\frac12\Bigl(I+\ta^1(\cos\ta^2\,\sigma_1+\sin\ta^2\,\sigma_2)\Bigr),
  \quad \ta^1\in(0,1),\ \ta^2\in[0,2\pi).
\]
The Bloch vector is
\[
  s_\ta=\begin{pmatrix}\ta^1\cos\ta^2\\ \ta^1\sin\ta^2\\ 0\end{pmatrix}.
\]
Assume an independent prior $\pi(\ta)=\pi_1(\ta^1)\pi_2(\ta^2)$, and denote
\[
  \mu_1=\mathbb{E}[\ta^1],\ \ v_1=\mathrm{Var}(\ta^1),\ \ 
  c=\mathbb{E}[\cos\ta^2],\ \ s=\mathbb{E}[\sin\ta^2],
\]
and (if needed)
\[
  \mu_2=\mathbb{E}[\ta^2],\ \ a_c=\mathbb{E}[\ta^2\cos\ta^2],\ \ a_s=\mathbb{E}[\ta^2\sin\ta^2].
\]
Then the Bayesian mean Bloch vector is
\[
  \ket{s_\mathrm{B}}=\mathbb{E}[s_\ta]=(\mu_1 c,\ \mu_1 s,\ 0)^\top.
\]
Moreover, the matrix $\Csf_\mathrm{B}\in\mathbb{R}^{2\times 3}$ defined by
$\Csf_\mathrm{B}^{ij}=\mathbb{E}[\ta^i s_\ta^j]-\mu^i s_\mathrm{B}^j$
is explicitly given by
\[
  \Csf_\mathrm{B}=
  \begin{pmatrix}
    v_1 c & v_1 s & 0\\
    \mu_1(a_c-\mu_2 c) & \mu_1(a_s-\mu_2 s) & 0
  \end{pmatrix},
\]
whose third column is identically zero. The antisymmetric matrix $\Fsf_\mathrm{B}$ defined from $\ket{s_\mathrm{B}}$ becomes
\[
  \Fsf_\mathrm{B}=
  \begin{pmatrix}
    0&0&\mu_1 s\\
    0&0&-\mu_1 c\\
    -\mu_1 s&\mu_1 c&0
  \end{pmatrix}.
\]
A direct multiplication shows
\[
  \Csf_\mathrm{B}\Fsf_\mathrm{B}\Csf_\mathrm{B}^\top=\mathbf{0},
  \qquad
  \Tr\big|\Wsf^{1/2}\Csf_\mathrm{B}\Fsf_\mathrm{B}\Csf_\mathrm{B}^\top\Wsf^{1/2}\big|=0.
\]
Therefore, by Theorem~\ref{thm:qubit_formula}, the maximum posterior variance bound always coincides with the Bayesian SLD posterior variance bound:
\[
  \Ccal_{\mathrm{B}}^{\max}(\Wsf)=\Ccal_{\mathrm{B}}^{(\la=0)}(\Wsf).
\]

\subsection{Appendix: Calculation details for Model C}\label{App:modelC}

We consider the three-parameter qubit model
\[
  \rho_\ta=\half\bigl(I+s_\ta\cdot\sigma\bigr),
  \qquad
  s_\ta=\begin{pmatrix}\ta^1\\ \ta^2\\ \ta^3\end{pmatrix},
  \qquad |\ta|<1,
\]
together with an independent i.i.d. prior
\[
  \pi(\ta)=\pi_1(\ta^1)\pi_2(\ta^2)\pi_3(\ta^3),
  \qquad \pi_1=\pi_2=\pi_3.
\]
Then the mean vector and covariance matrix of $\ta$ are
\[
  \ket{\mu}:=\int_\Theta d\ta\,\pi(\ta)\,\ta,\qquad
  \Csf_\pi:=\int_\Theta d\ta\,\pi(\ta)\,(\ta-\mu)(\ta-\mu)^\top=v\,\Isf_3,\quad v>0.
\]
Since $s_\ta=\ta$, the Bayesian mean Bloch vector is
\[
  \ket{s_\mathrm{B}}:=\int_\Theta d\ta\,\pi(\ta)\,s_\ta=\ket{\mu},
  \qquad \|s_\mathrm{B}\|<1,
\]
and the matrix $\Csf_\mathrm{B}$ defined by
\[
  \Csf_\mathrm{B}^{ij}:=\int_\Theta d\ta\,\pi(\ta)\,\ta^i s_\ta^j-\mu^i s_\mathrm{B}^j
\]
reduces to the covariance:
\[
  \Csf_\mathrm{B}=\Csf_\pi=v\,\Isf_3.
\]

\paragraph{The antisymmetric matrix $\Fsf_\mathrm{B}$.}
From $\ket{s_\mathrm{B}}=(s_\mathrm{B}^1,s_\mathrm{B}^2,s_\mathrm{B}^3)^\top$, we define
\[
  \Fsf_\mathrm{B}=
  \begin{pmatrix}
    0 & -s_{\mathrm{B}}^{3} & s_{\mathrm{B}}^{2} \\
    s_{\mathrm{B}}^{3} & 0 & -s_{\mathrm{B}}^{1} \\
    -s_{\mathrm{B}}^{2} & s_{\mathrm{B}}^{1} & 0 
  \end{pmatrix}.
\]
It satisfies $\Fsf_\mathrm{B}^\top=-\Fsf_\mathrm{B}$ and $\Fsf_\mathrm{B}\ket{s_\mathrm{B}}=\mathbf{0}$.
Moreover, $\Fsf_\mathrm{B}$ has singular values $\{\|s_\mathrm{B}\|,\|s_\mathrm{B}\|,0\}$, hence
\begin{equation}
  \Tr|\Fsf_\mathrm{B}|=2\|s_\mathrm{B}\|.
  \label{eq:AppC_trF}
\end{equation}

\paragraph{The key trace-norm term in Theorem~\ref{thm:qubit_formula}.}
We choose the isotropic weight matrix $\Wsf=\Isf_3$ (equivalently, any $\Wsf\propto\Isf_3$ yields the same phase classification). Then
\[
  \Wsf^{\frac12}\Csf_\mathrm{B}\Fsf_\mathrm{B}\Csf_\mathrm{B}^\top\Wsf^{\frac12}
  =\Csf_\mathrm{B}\Fsf_\mathrm{B}\Csf_\mathrm{B}^\top
  =v^2\Fsf_\mathrm{B},
\]
and by \eqref{eq:AppC_trF},
\begin{equation}
  \Tr\Bigl|\Wsf^{\frac12}\Csf_\mathrm{B}\Fsf_\mathrm{B}\Csf_\mathrm{B}^\top\Wsf^{\frac12}\Bigr|
  =v^2\Tr|\Fsf_\mathrm{B}|
  =2v^2\|s_\mathrm{B}\|.
  \label{eq:AppC_trabs}
\end{equation}
In particular, the first case of Theorem~\ref{thm:qubit_formula} occurs only at $s_\mathrm{B}=0$.

\paragraph{Evaluation of $\Delta_\mathrm{B}(\Wsf)$.}
Let
\[
  P_\perp:=\Isf_3-\frac{\ket{s_\mathrm{B}}\bra{s_\mathrm{B}}}{\|s_\mathrm{B}\|^2}
  \qquad (\|s_\mathrm{B}\|>0).
\]
Since $\Csf_\mathrm{B}^\top\Wsf\,\Csf_\mathrm{B}=v^2\Isf_3$, we have
\[
  \Tr\!\left[P_\perp\,\Csf_\mathrm{B}^\top\Wsf\,\Csf_\mathrm{B}\right]
  =v^2\Tr[P_\perp]
  =2v^2.
\]
Combining with \eqref{eq:AppC_trabs} gives
\begin{equation}
  \Delta_\mathrm{B}(\Isf_3)
  =\frac{\Tr\!\left[P_\perp\,\Csf_\mathrm{B}^\top\Wsf\,\Csf_\mathrm{B}\right]}
         {\Tr\Bigl|\Wsf^{\frac12}\Csf_\mathrm{B}\Fsf_\mathrm{B}\Csf_\mathrm{B}^\top\Wsf^{\frac12}\Bigr|}
  =\frac{1}{\|s_\mathrm{B}\|},
  \qquad (\|s_\mathrm{B}\|>0).
  \label{eq:AppC_Delta}
\end{equation}
Hence the condition in Theorem~\ref{thm:qubit_formula},
\[
  \Delta_\mathrm{B}(\Wsf)\le \frac12\left(1+\frac{1}{\|s_\mathrm{B}\|^2}\right),
\]
is always satisfied for $\Delta_\mathrm{B}(\Isf_3)=1/\|s_\mathrm{B}\|$, since it is equivalent to
\[
  \left(1-\frac{1}{\|s_\mathrm{B}\|}\right)^2\ge 0.
\]
Therefore, for $\|s_\mathrm{B}\|>0$ we are always in the second case of Theorem~\ref{thm:qubit_formula}, and the optimal choice is $\la=\pm1$.
For $s_\mathrm{B}=0$, the bound is $\la$-independent by Theorem~\ref{thm:qubit_formula}.

\paragraph{Explicit form of $\Ksf_\mathrm{B}^{(\la)}$ at $\la=1$.}
In this isotropic setting, the matrix $\Ksf_\mathrm{B}^{(\la)}$ appearing in the Bayesian $\lambda$-posterior variance bound admits the following expression at $\la=1$:
\begin{equation}
  \Ksf_\mathrm{B}^{(\la=1)}
  =\ketbra{\mu}{\mu}+\frac{v^2}{1-\|s_\mathrm{B}\|^2}\bigl(\Isf_3-i\Fsf_\mathrm{B}\bigr),
\end{equation}
where $\ket{\mu}=\ket{s_\mathrm{B}}$ in the present model.
Substituting this into the general formula for $\Ccal_\mathrm{B}^{(\la)}(\Wsf)$ yields the closed form
$\Ccal_\mathrm{B}^{\max}=\Ccal_\mathrm{B}^{(\la=\pm1)}$ in the main text.

\end{document}